\documentclass[11pt]{article}

\usepackage[T1]{fontenc} %
\usepackage{lmodern} %
\usepackage[latin1]{inputenc} %
\usepackage[margin=1in]{geometry} %
\usepackage[pdftex]{graphicx} %
\usepackage{caption} %
\usepackage{subcaption} %
\usepackage{amsmath,amssymb,amsthm} %
\usepackage{paralist} %
\usepackage{xspace} %
\usepackage[absolute,overlay]{textpos} %
\usepackage{tikz} %
\usepackage[english]{babel} %
\usepackage{textcomp} %
\usepackage{refcount} %
\usepackage{hyperref} %
\usepackage{array} %
\usepackage{rotating} %
\usepackage{slashbox} %
\usepackage{refcount} %
\usepackage{bold-extra} %

\usepackage[backgroundcolor=white,linecolor=black,textsize=scriptsize]{todonotes} %
\theoremstyle{plain} %
\newcounter{thmcounter} %
\newtheorem{theorem}[thmcounter]{Theorem} %
\newtheorem{lemma}{Lemma} %
\newtheorem{claim}{Claim} %
\newtheorem{corollary}{Corollary} %
\theoremstyle{definition} %

\newenvironment{rethm}[1]{ %
  \newcounter{tmpcounter} %
  \setcounter{tmpcounter}{\thethmcounter} 
  \setcounterref{thmcounter}{#1} %
  \addtocounter{thmcounter}{-1}
  \begin{theorem}} { %
  \end{theorem} %
  \setcounter{thmcounter}{\thetmpcounter}}

\setdefaultitem{\textopenbullet}{}{}{}

\newcommand{\eps}{\varepsilon}
\newcommand{\eref}{\ensuremath{e_\mathrm{ref}}}

\DeclareMathOperator{\skel}{skel} %
\DeclareMathOperator{\pert}{pert} %
\DeclareMathOperator{\flex}{flex} %
\DeclareMathOperator{\cost}{cost} %
\DeclareMathOperator{\COST}{COST} %
\DeclareMathOperator{\rot}{rot} %
\DeclareMathOperator{\lef}{left} %
\DeclareMathOperator{\righ}{right} %
\DeclareMathOperator{\dem}{dem} %
\DeclareMathOperator{\flow}{flow} %

\begin{document}

\title{Optimal Orthogonal Graph Drawing with Convex Bend
  Costs\thanks{Part of this work was done within GRADR -- EUROGIGA
    project no. 10-EuroGIGA-OP-003.}}

\author{Thomas Bläsius, Ignaz Rutter, Dorothea Wagner}

\date{Karlsruhe Institute of Technology (KIT)\\%
\texttt{firstname.lastname@kit.edu}}

\maketitle

\begin{abstract}
  Traditionally, the quality of orthogonal planar drawings is
  quantified by either the total number of bends, or the maximum
  number of bends per edge.  However, this neglects that in typical
  applications, edges have varying importance.  Moreover, as bend
  minimization over all planar embeddings is $\mathcal {NP}$-hard,
  most approaches focus on a fixed planar embedding.

  We consider the problem {\sc OptimalFlexDraw} that is defined as
  follows.  Given a planar graph $G$ on $n$ vertices with maximum
  degree~4 and for each edge~$e$ a cost function $\cost_e \colon
  \mathbb N_0 \longrightarrow \mathbb R$ defining costs depending on
  the number of bends on $e$, compute an orthogonal drawing of $G$ of
  minimum cost.  Note that this optimizes over all planar embeddings
  of the input graphs, and the cost functions allow fine-grained
  control on the bends of edges.

  In this generality {\sc OptimalFlexDraw} is $\mathcal {NP}$-hard.
  We show that it can be solved efficiently if 1) the cost function of
  each edge is convex and 2) the first bend on each edge does not
  cause any cost (which is a condition similar to the positive
  flexibility for the decision problem {\sc FlexDraw}).  Moreover, we
  show the existence of an optimal solution with at most three bends
  per edge except for a single edge per block (maximal biconnected
  component) with up to four bends.  For biconnected graphs we obtain
  a running time of $\mathcal O(n\cdot T_{\flow}(n))$, where
  $T_{\flow}(n)$ denotes the time necessary to compute a minimum-cost
  flow in a planar flow network with multiple sources and sinks.  For
  connected graphs that are not biconnected we need an additional
  factor of $\mathcal O(n)$.
\end{abstract}

\section{Introduction}
\label{sec:introduction}

Orthogonal graph drawing is one of the most important techniques for
the human-readable visualization of complex data. Its \ae{}sthetic
appeal derives from its simplicity and straightforwardness. Since
edges are required to be straight orthogonal lines---which
automatically yields good angular resolution and short links---the
human eye may easily adapt to the flow of an edge.
The readability of orthogonal drawings can be further enhanced in the
absence of crossings, that is if the underlying data exhibits planar
structure.  Unfortunately, not all planar graphs have an orthogonal
drawing in which each edge may be represented by a straight horizontal
or vertical line. In order to be able to visualize all planar graphs
nonetheless, we allow edges to have bends. Since bends obfuscate the
readability of orthogonal drawings, however, we are interested in
minimizing the number of bends on the edges.

In this paper we consider the problem {\sc OptimalFlexDraw} whose
input consists of a planar graph $G$ with maximum degree~4 and for
each edge~$e$ a cost function $\cost_e \colon \mathbb N_0
\longrightarrow \mathbb R$ defining costs depending on the number of
bends on $e$.  We seek an orthogonal drawing of $G$ with minimum cost.
Garg and Tamassia~\cite{gt-curpt-01} show that it is $\mathcal
{NP}$-hard to decide whether a 4-planar graph admits an orthogonal
drawing without any bends.  Note that this directly implies that {\sc
  OptimalFlexDraw} is $\mathcal {NP}$-hard in general.  For a special
case, namely planar graphs with maximum degree~3 and series-parallel
graphs, Di Battista et al.~\cite{blv-sood-98} give an algorithm
minimizing the total number of bends optimizing over all planar
embeddings.  They introduce the concept of spirality that is similar to
the rotation we use (see Section~\ref{sec:orth-repr} for a
definition).  Bläsius et al.~\cite{bkrw-ogdfc-10} show that the
existence of a planar 1-bend drawing can be tested efficiently.  More
generally, they consider the problem {\sc FlexDraw}, where each edge
has a \emph{flexibility} specifying its allowed number of bends.  For
the case that all flexibilities are positive, they give a
polynomial-time algorithm for testing the existence of a valid
drawing.

As minimizing the number of bends for 4-planar orthogonal drawings is
$\mathcal{NP}$-hard, many results use the topology-shape-metrics
approach initially fixing the planar embedding.
Tamassia~\cite{t-eggmb-87} describes a flow network for minimizing the
number of bends.  This flow network can be easily adapted to also
solve {\sc OptimalFlexDraw} even for the case where the first bend may
cause cost, however, the planar embedding has to be fixed in advanced.
Biedl and Kant~\cite{bk-bhogd-98} show that every plane graph can be
embedded with at most two bends per edge except for the octahedron.
Morgana et al.~\cite{mps-aobepg-04} give a characterization of plane
graphs that have an orthogonal drawing with at most one bend per edge.
Tayu et al.~\cite{tnu-2dodspg-09} show that every series-parallel
graph can be drawn with at most one bend per edge.  All these results
and the algorithm we present here have the requirement of maximum
degree~4 in common.  Although this is a strong restriction it is
important to consider this case since algorithms dealing with
higher-degree vertices (drawing them as boxes instead of single
points) rely on algorithms for graphs with maximum
degree~4~\cite{tdb-agdrd-88, fkdhdglbn-95, km-qodpg-98}.

Even though fixing an embedding allows to efficiently minimize the
total number of bends (with this embedding), this neglects that the
choice of a planar embedding may have a huge impact on the number of
bends in the resulting drawing.  The result by Bläsius et
al.~\cite{bkrw-ogdfc-10} concerning the problem {\sc FlexDraw} takes
this into account and additionally allows the user to control the
final drawing, for example by allowing few bends on important edges.
However, if such a drawing does not exist, the algorithm solving {\sc
  FlexDraw} does not create a drawing at all and thus it cannot be
used in a practical application.  Thus, the problem {\sc
  OptimalFlexDraw}, which generalizes the corresponding optimization
problem, is of higher practical interest, as it allows the user to
take control of the properties of the final drawing within the set of
feasible drawings.  Moreover, it allows a more fine-grained control of
the resulting drawing by assigning high costs to bends on important
edges.

\paragraph{Contribution and Outline.}
\label{sec:contribution-outline}

Our main result is the first polynomial-time bend-optimization
algorithm for general 4-planar graphs 
optimizing over all embeddings.  Previous work considers only restricted graph classes and unit costs.
We solve {\sc OptimalFlexDraw} if 1) all cost functions
are convex and 2) the first bend is for free.   We
note that convexity is indeed quite natural, and that without
condition 2) {\sc OptimalFlexDraw} is $\mathcal{NP}$-hard, as it could
be used to minimize the total number of bends over all embeddings,
which is known to be $\mathcal{NP}$-hard~\cite{gt-curpt-01}. 

In particular, our algorithm allows to efficiently minimize the total
number of bends over all planar embeddings, where one bend per edge is
free.  Note that this is an optimization version of {\sc FlexDraw}
where each edges has flexibility~1, as a drawing with cost~0 exists if
and only if {\sc FlexDraw} has a valid solution.  Moreover, as it is
known that every 4-planar graph has an orthogonal representation with
at most two bends per edge~\cite{bk-bhogd-98}, our result can also be
used to create such a drawing minimizing the number of edges having
two bends by setting the costs for three or more bends to~$\infty$.

To derive the algorithm for {\sc OptimalFlexDraw}, we show the
existence of an optimal solution with at most three bends per edge
except for a single edge per block with up to four bends, confirming a
conjecture of Rutter~\cite{r-tmp-11}.

Our strategy for solving {\sc OptimalFlexDraw} for biconnected graphs
optimizing over all planar embedding is the following.  We use dynamic
programming on the SPQR-tree of the graph, which is a data structure
representing all planar embeddings of a biconnected graph.  Every node
in the SPQR-tree corresponds to a split component and we compute cost
functions for these split components determining the cost depending on
how strongly the split component is bent.  We compute such a cost
function from the cost functions of the children using a flow network
similar to the one described by Tamassia~\cite{t-eggmb-87}.  As
computing flows with minimum cost is $\mathcal {NP}$-hard for
non-convex costs we need to ensure that not only the cost functions of
the edges but also the cost functions of the split components we
compute are convex.  However, this is not true at all, see
Figure~\ref{fig:not-convex-example} for an example.  This is not even
true if every edge can have a single bend for free and then has to pay
cost~1 for every additional bend, see
Figure~\ref{fig:not-convex-example}(c).  To solve this problem, we
essentially show that it is sufficient to compute the cost functions
on the small interval $[0, 3]$.  We can then show that the cost
functions we compute are always convex on this interval.

\begin{figure}
  \centering
  \includegraphics[page=2]{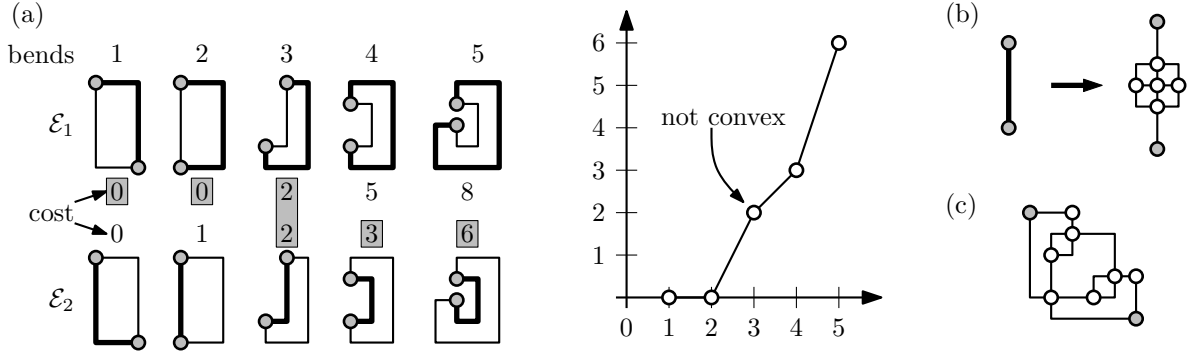}
  \caption{(a) Two parallel edges, the thin has one bend for free,
    every additional bend costs~1, the thick edge has two bends for
    free, every additional bend costs~2.  Whether embedding $\mathcal
    E_1$ or $\mathcal E_2$ is better depends on the number of bends.
    The minimum (marked by gray boxes) yields a non-convex cost
    function.  (b) The non-convexity in (a) does not rely on multiple
    edges, the thick edge could be replaced by the shown gadget where
    each edge of the gadget has one bend for free and every additional
    bend costs~2.  (c) This example has a non-convex cost function
    even if every edge has one bend for free and each additional bend
    costs~1.}
  \label{fig:not-convex-example}
\end{figure}

We start with some preliminaries in Section~\ref{sec:preliminaries}.
Afterwards, we first consider the decision problem {\sc FlexDraw} for
the case that the planar embedding is fixed in
Section~\ref{sec:valid-drawings-with}.  In this restricted setting we
are able to prove the existence of valid drawings with special
properties.  Bläsius et al.~\cite{bkrw-ogdfc-10} show that ``rigid''
graphs do not exist in this setting in the sense that a drawing that
is bent strongly can be unwound under the assumption that the
flexibility of every edge is at least~1.  In other words this shows
that graphs with positive flexibility behave similar to single edges
with positive flexibility.  We present a more elegant proof yielding a
stronger result that can then be used to reduce the number of bends of
every edge down to three (at least for biconnected graphs and except
for a single edge on the outer face).  In Section~\ref{sec:flex-split-comp} we
extend the term ``bends'', originally defined for edges, to split
components and show that in a biconnected graph the split components
corresponding to the nodes in its SPQR-tree can be assumed to have
only up to three bends.  In Section~\ref{sec:optim-draw-with} we show
that these results for the decision problem {\sc FlexDraw} can be
extended to the optimization problem {\sc OptimalFlexDraw}.  With this
result we are able to drop the fixed planar embedding
(Section~\ref{sec:opt-draw-var-emb}).  We first consider biconnected
graphs in Section~\ref{sec:biconnected-graphs} and compute cost
functions on the interval $[0, 3]$, which can be shown to be convex on
that interval, bottom up in the SPQR-tree.  In
Section~\ref{sec:connected-graphs} we extend this result to connected
graphs using the BC-tree (see Section~\ref{sec:spqr-tree} for a
definition).

\section{Preliminaries}
\label{sec:preliminaries}

In this section we introduce some notations and preliminaries.

\subsection{FlexDraw}
\label{sec:sc-flexdraw}

The original {\sc FlexDraw} problem asks for a given 4-planar graph $G
= (V, E)$ with a function $\flex \colon E \longrightarrow \mathbb N_0 \cup
\{\infty\}$ assigning a \emph{flexibility} to every edge whether an
orthogonal drawing of $G$ exists such that every edge $e \in E$ has at
most $\flex(e)$ bends.  Such a drawing is called a \emph{valid}
drawing of the {\sc FlexDraw} instance.  The problem \emph{\sc
  OptimalFlexDraw} is the optimization problem corresponding to the
decision problem {\sc FlexDraw} and is defined as follows.  Let $G =
(V, E)$ be a 4-planar graph together with a cost function $\cost_e :
\mathbb N_0 \longrightarrow \mathbb R \cup \{\infty\}$ associated with
every edge $e \in E$ having the interpretation that $\rho$ bends on
the edge $e$ cause $\cost_e(\rho)$ cost.  Then the \emph{cost of an
  orthogonal drawing} of $G$ is the total cost summing over all edges.
A drawing is \emph{optimal} if it has the minimum cost among all
orthogonal drawings of $G$.  The task of the optimization problem {\sc
  OptimalFlexDraw} is to find an optimal drawing of $G$.

Since {\sc OptimalFlexDraw} contains the $\mathcal {NP}$-hard problem
{\sc FlexDraw}, it is $\mathcal {NP}$-hard itself.  However, {\sc
  FlexDraw} is efficiently solvable for instances with \emph{positive
  flexibility}, that is instances in which the flexibility of every
edge is at least~1.  To obtain a similar result for {\sc
  OptimalFlexDraw} we have to restrict the possible cost functions
slightly.

For a cost function $\cost_e(\cdot)$ we define the \emph{difference
  function} $\Delta\cost_e(\cdot)$ to be $\Delta\cost_e(\rho) =
\cost_e(\rho + 1) - \cost_e(\rho)$.  A cost function is
\emph{monotone} if its difference function is greater or equal to~0.
We say that the \emph{base cost} of the edge $e$ with monotone cost
function is $b_e = \cost_e(0)$.  The \emph{flexibility} of an edge $e$
with monotone cost function is defined to be the largest possible
number of bends $\rho$ for which $\cost_e(\rho) = b_e$.  As before, we
say that an instance $G$ of {\sc OptimalFlexDraw} has \emph{positive
  flexibility} if all cost functions are monotone and the flexibility
of every edge is positive.  Unfortunately, we have to restrict the
cost functions further to be able to solve {\sc OptimalFlexDraw}
efficiently.  The cost function $\cost_e(\cdot)$ is \emph{convex}, if
its difference function is monotone.  We call an instance of {\sc
  OptimalFlexDraw} \emph{convex}, if every edge has positive
flexibility and each cost function is convex.  Note that this includes
that the cost functions are monotone.  We provide an efficient
algorithm solving {\sc OptimalFlexDraw} for convex instances.

\subsection{Connectivity, BC-Tree and SPQR-Tree}
\label{sec:spqr-tree}

A graph is \emph{connected} if there exists a path between any pair of
vertices.  A \emph{separating $k$-set} is a set of $k$ vertices whose
removal disconnects the graph.  Separating 1-sets and 2-sets are
\emph{cutvertices} and \emph{separation pairs}, respectively.  A
connected graph is \emph{biconnected} if it does not have a cut vertex
and \emph{triconnected} if it does not have a separation pair.  The
maximal biconnected components of a graph are called \emph{blocks}.
The \emph{cut components} with respect to a separation $k$-set $S$ are
the maximal subgraphs that are not disconnected by removing $S$.

The \emph{block-cutvertex tree (BC-tree)} $\mathcal B$ of a connected
graph is a tree whose nodes are the blocks and cutvertices of the
graph, called \emph{B-nodes} and \emph{C-nodes}, respectively.  In the
BC-tree a block $B$ and a cutvertex $v$ are joined by an edge if $v$
belongs to $B$.  If an embedding is chosen for each block, these
embeddings can be combined to an embedding of the whole graph if and
only if $\mathcal B$ can be rooted at a B-node such that the parent of
every other block $B$ in $\mathcal B$, which is a cutvertex, lies on
the outer face of $B$.

We use the \emph{SPQR-tree} introduced by Di Battista and
Tamassia~\cite{dt-omtc-96,dt-opt-96} to represent all planar
embeddings of a biconnected planar graph $G$.  The SPQR-tree $\mathcal
T$ of $G$ is a decomposition of $G$ into its triconnected components
along its \emph{split pairs} where a split pair is either a separation
pair or an edge.  We first define the SPQR-tree to be unrooted,
representing embeddings on the sphere, that is planar embeddings
without a designated outer face.  Let $\{s, t\}$ be a split pair and
let $H_1$ and~$H_2$ be two subgraphs of $G$ such that $H_1 \cup H_2 =
G$ and $H_1 \cap H_2 = \{s, t\}$.  Consider the tree containing the
two nodes $\mu_1$ and $\mu_2$ associated with the graphs $H_1 + \{s,
t\}$ and $H_2 + \{s, t\}$, respectively.  These graphs are called
\emph{skeletons} of the nodes $\mu_i$, denoted by $\skel(\mu_i)$ and
the special edge $\{s, t\}$ is said to be a \emph{virtual edge}.  The
two nodes $\mu_1$ and $\mu_2$ are connected by an edge, or more
precisely, the occurrence of the virtual edges $\{s, t\}$ in both
skeletons are linked by this edge.  Now a combinatorial embedding of
$G$ uniquely induces a combinatorial embedding of $\skel(\mu_1)$ and
$\skel(\mu_2)$.  Furthermore, arbitrary and independently chosen
embeddings for the two skeletons determine an embedding of $G$, thus
the resulting tree can be used to represent all embeddings of $G$ by
the combination of all embeddings of two smaller planar graphs.  This
replacement can of course be applied iteratively to the skeletons
yielding a tree with more nodes but smaller skeletons associated with
the nodes.  Applying this kind of decomposition in a systematic way
yields the SPQR-tree as introduced by Di Battista and
Tamassia~\cite{dt-omtc-96,dt-opt-96}.  The SPQR-tree $\mathcal T$ of a
biconnected planar graph $G$ contains four types of nodes.  First, the
P-nodes having a bundle of at least three parallel edges as skeleton
and a combinatorial embedding is given by any ordering of these edges.
Second, the skeleton of an R-node is triconnected, thus having exactly
two embeddings~\cite{w-nspg-32}, and third, S-nodes have a simple
cycle as skeleton without any choice for the embedding.  Finally,
every edge in a skeleton representing only a single edge in the
original graph $G$ is formally also considered to be a virtual edge
linked to a Q-node in $\mathcal T$ representing this single edge.
Note that all leaves of the SPQR-tree~$\mathcal T$ are Q-nodes.
Besides from being a nice way to represent all embeddings of a
biconnected planar graph, the SPQR-tree has only size linear in $G$
and Gutwenger and Mutzel~\cite{gm-lti-00} showed how to compute it in
linear time.  Figure~\ref{fig:spqr-tree} shows a biconnected planar
graph together with its SPQR-tree.

\begin{figure}
  \centering
  \includegraphics[page=1]{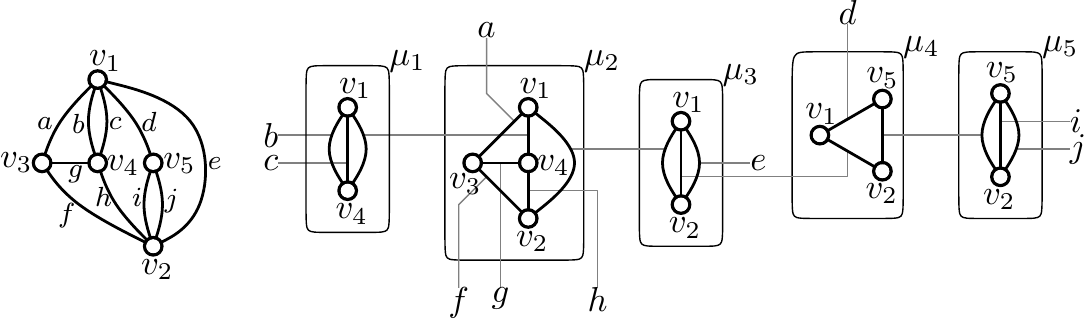}
  \caption{The unrooted SPQR-tree of a biconnected planar graph.  The
    nodes $\mu_1$, $\mu_3$ and $\mu_5$ are P-nodes, $\mu_2$ is an
    R-node and $\mu_4$ is an S-node.  The Q-nodes are not shown
    explicitely.}
  \label{fig:spqr-tree}
\end{figure}

Often the SPQR-tree $\mathcal T$ of a biconnected planar graph $G$ is
assumed to be rooted in a Q-node representing all planar embeddings
with the corresponding edge on the outer face.  In contrast to
previous results, we assume the SPQR-tree $\mathcal T$ to be rooted in
some node $\tau$, which may be a Q-node or an inner node.  In the
following we describe the interpretation of the SPQR-tree with root
$\tau$.  Every node $\mu$, apart form~$\tau$ itself, has a unique
parent and thus its skeleton $\skel(\mu)$ contains a virtual edge
corresponding to this parent.  We refer to this virtual edge as the
\emph{parent edge}.  A planar embedding $\mathcal E$ of $G$ is
represented by $\mathcal T$ with root $\tau$ if the embedding induced
on the skeleton $\skel(\mu)$ of every node $\mu \not= \tau$ has the
parent edge on the outer face.  The embedding of $\skel(\tau)$ is not
restricted, thus the choice of the outer face makes a difference for
the root.

For every node $\mu$ in the SPQR-tree $\mathcal T$ apart from the root
$\tau$ we define the \emph{pertinent graph} of~$\mu$, denoted by
$\pert(\mu)$, as follows.  The pertinent graph of a Q-node is the edge
associated to it.  The pertinent graph of an inner node $\mu$ is
recursively defined to be the graph obtained by replacing all virtual
edges apart from the parent edge by the pertinent graphs of the
corresponding children in~$\mathcal T$.  The \emph{expansion graph} of
a virtual edge $\eps$ in $\skel(\mu)$ is the pertinent graph of $\mu'$
where $\mu'$ is the child of $\mu$ corresponding to the virtual edge
$\eps$ with respect to the root $\mu$.

\subsection{Orthogonal Representation}
\label{sec:orth-repr}

Two orthogonal drawings of a 4-planar graph $G$ are \emph{equivalent},
if they have the same topology, that is the same planar embedding, and
the same shape in the sense that the sequence of right and left turns
is the same in both drawings when traversing the faces of $G$.  To
make this precis, we define \emph{orthogonal representations},
originally introduced by Tamassia~\cite{t-eggmb-87}, as equivalence
classes of this equivalence relation between orthogonal drawings.  To
ease the notation we first only consider the biconnected case.

Let $\Gamma$ be an orthogonal drawing of a biconnected 4-planar graph
$G$.  In the planar embedding $\mathcal E$ induced by $\Gamma$ every
edge $e$ is incident to two different faces, let $f$ be one of them.
When traversing $f$ in clockwise order (counter-clockwise if $f$ is
the outer face) $e$ may have some bends to the right and some bends to
the left.  We define the \emph{rotation} of $e$ in the face $f$ to be
the number of bends to the right minus the number of bends to the left
and denote the resulting value by $\rot(e_f)$.  Similarly, every
vertex $v$ is incident to several faces, let $f$ be one of them.  Then
we define the rotation of $v$ in $f$, denoted by $\rot(v_f)$, to be
$1$, $-1$ and $0$ if there is a turn to the right, a turn to the left
and no turn, respectively, when traversing~$f$ in clockwise direction
(counter-clockwise if $f$ is the outer face).  The orthogonal
representation $\mathcal R$ belonging to $\Gamma$ consists of the
planar embedding $\mathcal E$ of $G$ and all rotation values of edges
and vertices, respectively.  It is easy to see that every orthogonal
representation has the following properties.
\begin{compactenum}[(I)]
\item For every edge $e$ incident to the faces $f_1$ and $f_2$ the
  equation $\rot(e_{f_1}) = -\rot(e_{f_2})$ holds.
\item The sum over all rotations in a face is $4$ for inner faces and
  $-4$ for the outer face.
\item The sum of rotations around a vertex $v$ is $2 \cdot (\deg(v) -
  2)$.
\end{compactenum}
Tamassia showed that the converse is also true~\cite{t-eggmb-87}, that
is $\mathcal R$ is an orthogonal representation representing a class
of orthogonal drawings if the rotation values satisfy the above
properties.  He moreover describes a flow network such that every flow
in the flow network corresponds to an orthogonal representation.  A
modification of this flow network can also be used to solve {\sc
  OptimalFlexDraw} but only for the case that the planar embedding is
fixed.  In some cases we also write $\rot_{\mathcal R}(\cdot)$ instead
of $\rot(\cdot)$ to make clear to which orthogonal representation we
refer to.  Moreover, the face in the index is sometimes omitted if it
is clear which face is meant.

When extending the term orthogonal representation to not necessarily
biconnected graphs there are two differences.  First, a vertex $v$
with $\deg(v) = 1$ may exist.  Then $v$ is incident to a single face
$f$ and we define the rotation $\rot(v_f)$ to be~$-2$.  Note that the
rotations around every vertex~$v$ still sum up to $2 \cdot (\deg(v) -
2)$.  The second difference is that the notation introduced above is
ambiguous since edges and vertices may occur several times in the
boundary of the same face.  For example a bridge $e$ is incident to
the face $f$ twice, thus it is not clear which rotation is meant by
$\rot(e_f)$.  However, it will always be clear from the context, which
incidence to the face $f$ is meant by the index $f$.  Thus, we use for
connected graphs the same notation as for biconnected graphs.

Let $G$ be a 4-planar graph with orthogonal representation $\mathcal
R$ and two vertices $s$ and $t$ incident to a common face $f$.  We
define $\pi_f(s, t)$ to be the unique shortest path from $s$ to $t$ on
the boundary of $f$, when traversing $f$ in clockwise direction
(counter-clockwise if $f$ is the outer face).  Let $s = v_1, \dots,
v_k = t$ be the vertices on the path $\pi_f(s, t)$.  The rotation of
$\pi(s, t)$ is defined as 
\[\rot(\pi(s, t)) = \sum_{i = 1}^{k-1} \rot(\{v_i, v_{i+1}\}) +
\sum_{i = 2}^{k-1} \rot(v_i)\,,\] 
where all rotations are with respect to the face $f$.

Note that it does not depend on the particular drawing of a graph $G$
how many bends each edge has but only on the orthogonal
representation.  Thus we can continue searching for valid and optimal
orthogonal representations instead of drawings to solve {\sc FlexDraw}
and {\sc OptimalFlexDraw}, respectively.

Let $G$ be a 4-planar graph with positive flexibility and valid
orthogonal representation $\mathcal R$ and let $\{s, t\}$ be a split
pair.  Let further $H$ be a split component with respect to $\{s, t\}$
such that the orthogonal representation $\mathcal S$ of $H$ induced by
$\mathcal R$ has $\{s, t\}$ on the outer face $f$.  The orthogonal
representation $\mathcal S$ of $H$ is called \emph{tight} with respect
to the vertices $s$ and $t$ if the rotations of $s$ and $t$ in
internal faces are~1, that is $s$ and $t$ form $90^\circ$-angles in
internal faces of $H$.  Bläsius et al.~\cite[Lemma~2]{bkrw-ogdfc-10}
show that $\mathcal S$ can be made tight with respect to $s$ and $t$,
that is there exists a valid tight orthogonal representation of $H$
that is tight.  Moreover, this tight orthogonal representation can be
plugged back into the orthogonal representation of the whole graph
$G$.  We call an orthogonal representation~$\mathcal R$ of the whole
graph $G$ \emph{tight}, if every split component having the
corresponding split pair on its outer face is tight with respect to
its split pair.  It follows that we can assume without loss of
generality that every valid orthogonal representation is tight.  This
has two major advantages.  First, if we have for example a chain of
graphs and orthogonal representations of each graph in the chain, we
can combine these orthogonal representations by simply stacking them
together; see Figure~\ref{fig:stacking-tight-graphs}.  Note that this
may not be possible if the orthogonal representations are not tight.
Second, the shape of the outer face~$f$ of a split component with
split pair $\{s, t\}$ is completely determined by the rotation of
$\pi_f(s, t)$ and the degrees of $s$ and $t$, since the rotation at
the vertices $s$ and $t$ in the outer face only depends on their
degrees.  In the following we assume every orthogonal representation
to be tight.

\begin{figure}
  \centering
  \includegraphics{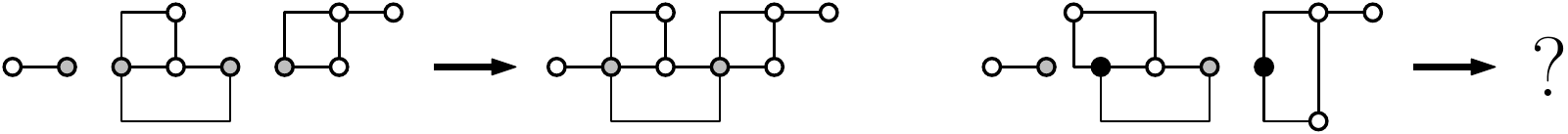}
  \caption{On the left three tight orthogonal drawings are stacked
    together.  This is not possible on the right side, since the black
    vertices have angles larger than $90^\circ$ in internal faces.}
  \label{fig:stacking-tight-graphs}
\end{figure}

\subsection{Flow Network}
\label{sec:flow-network}

A \emph{cost flow network} (or \emph{flow network} for short) is a
tuple $N = (V, A, \COST, \dem)$ where $(V, A)$ is a directed
(multi-)graph, $\COST$ is a set containing a \emph{cost function}
$\cost_a \colon \mathbb N_0 \longrightarrow \mathbb R \cup \{\infty\}$ for
each arc $a \in A$ and $\dem \colon V \longrightarrow \mathbb Z$ is the
\emph{demand} of the vertices.  A \emph{flow} in $N$ is a function
$\phi \colon A \longrightarrow \mathbb N_0$ assigning a certain amount of
flow to each arc.  A flow $\phi$ is \emph{feasible}, if the difference
of incoming and outgoing flow at each vertex equals its demand, that
is
\[\dem(v) = \sum_{(u,v) \in A} \phi(u,v) - \sum_{(v,u) \in A} \phi(v,u)
\,\, \text{ for all } v \in V.\]
The \emph{cost} of a given flow $\phi$ is the total cost of the arcs
caused by the flow $\phi$, that is
\[\cost(\phi) = \sum_{a \in A} \cost_a(\phi(a)).\]
A feasible flow $\phi$ in $N$ is called \emph{optimal} if $\cost(\phi)
\le \cost(\phi')$ holds for every feasible flow $\phi'$.

If the cost function of an arc $a$ is $0$ on an interval $[0, c]$ and
$\infty$ on $(c, \infty)$, we say that $a$ has \emph{capacity} $c$.

A flow network $N$ is called \emph{convex} if the cost functions on
its arcs are convex.  In the flow networks we consider, every arc $a
\in A$ has a corresponding arc $a' \in A$ between the same vertices
pointing in the opposite direction.  A flow $\phi$ is
\emph{normalized} if $\phi(a) = 0$ or $\phi(a') = 0$ for each of these
pairs.  Since we only consider convex flow networks a normalized
optimal flow does always exist.  Thus we assume without loss of
generality that all flows are normalized.  We simplify the notation as
follows.  If we talk about an amount of flow on the arc $a$ that is
negative, we instead mean the same positive amount of flow on the
opposite arc $a'$.  In many cases minimum-cost flow networks are only
considered for linear cost functions, that is each unit of flow on an
arc causes a constant cost defined for that arc.  Note that the cost
functions in a convex flow network $N$ are piecewise linear and convex
according to our definition.  Thus, it can be easily formulated as a
flow network with linear costs by splitting every arc into multiple
arcs, each having linear costs.  It is well known that flow networks
of this kind can be solved in polynomial time.  The best known running
time depends on additional properties that $N$ may satisfy.  We use an
algorithm computing a minimum-cost flow in the network $N$ as black
box and denote the necessary running time by $T_{\flow}(|N|)$.  In
Section~\ref{sec:computing-flow} we have a closer look on which
algorithm to use.

Let $u, v \in V$ be two nodes of the convex flow network $N$ with
demands $\dem(u)$ and $\dem(v)$.  The \emph{parameterized flow
  network} with respect to the nodes $u$ and $v$ is defined the same
as $N$ but with a \emph{parameterized demand} of $\dem(u) - \rho$ for
$u$ and $\dem(v) + \rho$ for $v$ where $\rho$ is a parameter.  The
\emph{cost function} $\cost_{N}(\rho)$ of the parameterized flow
network $N$ is defined to be $\cost(\phi)$ of an optimal flow $\phi$
in $N$ with respect to the parameterized demands determined by $\rho$.
Note that increasing $\rho$ by~1 can be seen as pushing one unit of
flow from $u$ to $v$.  We define the \emph{optimal parameter}~$\rho_0$
to be the parameter for which the cost function is minimal among all
possible parameters.  The correctness of the minimum weight path
augmentation method to compute flows with minimum costs implies the
following theorem~\cite{ek-tiaenfp-72}.

\begin{theorem}
  \label{thm:flow-has-convex-cost}
  The cost function of a parameterized flow network is convex on the
  interval~$[\rho_0, \infty]$, where $\rho_0$ is the optimal
  parameter.
\end{theorem}
\begin{proof}
  Let $N = (V, A, \COST, \dem)$ be a parameterized flow network and
  let $\phi_0$ be a minimum-cost flow in $N$ with respect to the
  optimal parameter $\rho_0$.  To simplify notation, we assume $\rho_0
  = 0$.  The \emph{residual network} $R_0$ with respect to $\phi_0$ is
  the graph $(V, A)$ with a constant cost $\cost_0(a)$ assigned to
  every arc $a$ such that $\cost_0(a)$ is the amount of cost in $N$
  that has to be payed to push an additional unit of flow along $a$,
  with respect to the given flow $\phi_0$.  Note that this cost may be
  negative.  It is well known that an optimal flow $\phi_1$ with
  respect to the parameter~1 can be computed by pushing one unit of
  flow along a path from $u$ to $v$ with minimum weight in
  $R_0$~\cite{ek-tiaenfp-72}.  Moreover, we can continue and compute
  an optimal flow $\phi_{k+1}$ by augmenting $\phi_k$ along a minimum
  weight path in the residual network $R_k$ with respect to the flow
  $\phi_k$.  Assume we augment $\phi_k$ along the path $\pi_k$ causing
  cost $\cost_k(\pi_k)$ to obtain an optimal flow $\phi_{k+1}$ with
  respect to the parameter $k+1$ and then we augment along a path
  $\pi_{k+1}$ in $R_{k+1}$ with cost $\cost_{k+1}(\pi_{k+1})$ to
  obtain an optimal flow $\phi_{k+2}$ with respect to the parameter
  $k+2$.  To obtain the claimed convexity we have to show that
  $\cost_{k}(\pi_k) \le \cost_{k+1}(\pi_{k+1})$ holds.

  If $\pi_k$ and $\pi_{k+1}$ contain an arc $a$ in the same direction,
  then $\cost_k(a) \le \cost_{k+1}(a)$ holds by the convexity of the
  cost function of $a$.  If $\pi_k$ contains the arc $a$ and
  $\pi_{k+1}$ contains the arc $a'$ in the opposite direction then
  $\cost_k(a) = -\cost_{k+1}(a')$ holds.  Assume $\pi_k$ and
  $\pi_{k+1}$ share such an arc in the opposite direction.  Then we
  remove this arc in both directions, splitting each of the paths
  $\pi_k$ and $\pi_{k+1}$ into two subpaths.  We define two new paths
  $\pi$ and $\pi'$ by concatenating the first part of $\pi_k$ with the
  second part of $\pi_{k+1}$ and vice versa, respectively.  This can
  be done iteratively, thus we can assume that $\pi$ and $\pi'$ do not
  share arcs in the opposite direction.  We consider the cost of $\pi$
  and $\pi'$ in the residual network $R_k$.  Obviously, for an arc $a$
  that is exclusively contained either in $\pi$ or in $\pi'$ we have
  $\cost_k(a) = \cost_{k+1}(a)$.  For an arc that is contained in
  $\pi$ and~$\pi'$ we have $\cost_k(a) \le \cost_{k+1}(a)$.  Moreover,
  for every pair of arcs $a$ and $a'$ that was removed we have
  $\cost_k(a) = -\cost_{k+1}(a')$.  This yields the inequality
  $\cost_k(\pi_k) + \cost_{k+1}(\pi_{k+1}) \ge \cost_k(\pi) +
  \cost_k(\pi')$.  Since $\pi_k$ was a path with smallest possible
  weight in $R_k$ we have $\cost_{k}(\pi_k) \le \cost_k(\pi)$ and
  $\cost_k(\pi_k) \le \cost_k(\pi')$.  With the above inequality this
  yields $\cost_{k+1}(\pi_{k+1}) \ge \cost_k(\pi_k)$.
\end{proof}

\section{Valid Drawings with Fixed Planar Embedding}
\label{sec:valid-drawings-with}

In this section we consider the problem {\sc FlexDraw} for the case
that the planar embedding is fixed.  We show that the existence of a
valid orthogonal representation implies the existence of a valid
orthogonal representation with special properties.  We first show the
following.  Given a biconnected 4-planar graph with positive
flexibility and an orthogonal representation $\mathcal R$ such that
two vertices~$s$ and $t$ lie on the outer face $f$, then the rotation
along $\pi_f(s, t)$ can be reduced by~1 if it is at least~0.  This
result is a key observation for the algorithm solving the decision
problem {\sc FlexDraw}~\cite{bkrw-ogdfc-10}.  It in a sense shows that
``rigid'' graphs that have to bent strongly do not exists.  This kind
of graphs play an important role in the ${\mathcal NP}$-hardness proof
of 0-embeddability by Garg and Tamassia~\cite{gt-curpt-01}.  Moreover,
we show the existence of a valid orthogonal representation $\mathcal
R'$ inducing the same planar embedding and having the same angles
around vertices as $\mathcal R$ such that every edge has at most three
bends in $\mathcal R'$, except for a single edge on the outer face
with up to five bends.  If we allow to change the embedding slightly,
this special edge has only up to four bends.

Let $G$ be a 4-planar graph with positive flexibility and valid
orthogonal representation $\mathcal R$, and let $e$ be an edge.  If
the number of bends of $e$ equals its flexibility, we orient $e$ such
that its bends are right bends.  Otherwise, $e$ remains undirected.
We define a path $\pi = (v_1, \dots, v_k)$ in $G$ to be a
\emph{directed path}, if the edge $\{v_i, v_{i+1}\}$ (for $i \in \{1,
\dots, k-1\}$) is either undirected or directed from~$v_i$ to
$v_{i+1}$.  A path containing only undirected edges can be seen as
directed path for both possible directions.  The path $\pi$ is
\emph{strictly directed}, if it is directed and does not contain
undirected edges.  These terms directly extend to \emph{(strictly)
  directed cycles}.  Given a (strictly) directed cycle $C$ the terms
$\lef(C)$ and $\righ(C)$ denote the set of edges and vertices of $G$
lying to the left and right of~$C$, respectively, with respect to the
orientation of $C$.  A cut $(U, V \setminus U)$ is said to be
\emph{directed} from~$U$ to $V \setminus U$, if every edge $\{u, v\}$
with $u \in U$ and $v \in V \setminus U$ is either directed from $u$
to $v$ or undirected.  According to the above definitions a cut is
\emph{strictly directed} from $U$ to $V \setminus U$ if it is directed
and contains no undirected edges.  Before we show how to unwind an
orthogonal representation that is bent strongly we need the following
technical lemma.

\begin{figure}
  \centering
  \includegraphics{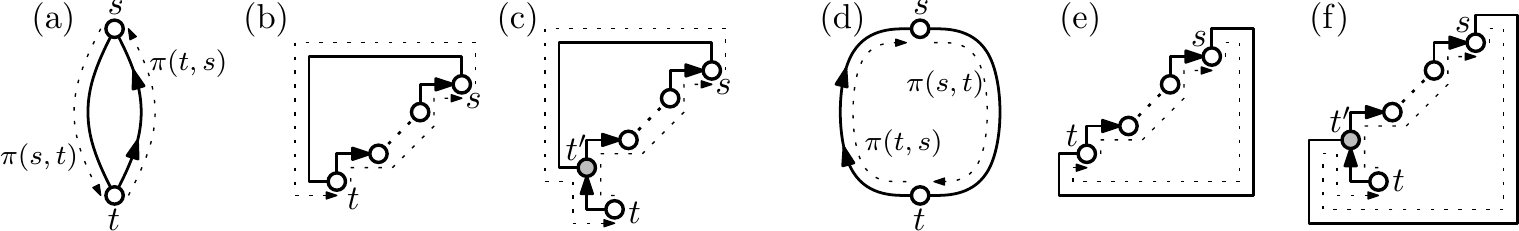}
  \caption{Since a strictly directed path from $t$ to $s$ has a lower
    bound for its rotation this yields upper bounds for paths from $s$
    to $t$ (Lemma~\ref{lem:strictly-directed-path}).}
  \label{fig:strictly-directed-path}
\end{figure}

\begin{lemma}
  \label{lem:strictly-directed-path}
  Let $G$ be a graph with positive flexibility and vertices $s$ and
  $t$ such that $G + st$ is biconnected and 4-planar.  Let further
  $\mathcal R$ be a valid orthogonal representation with $s$ and $t$
  incident to the common face $f$ such that $\pi_f(t, s)$ is strictly
  directed from $t$ to $s$. Then the following holds.
  \begin{compactenum}[(1)]
  \item $\rot_{\mathcal R}(\pi_f(s, t)) \le -3$ if $f$ is the outer
    face and $G$ does not consist of a single path
  \item $\rot_{\mathcal R}(\pi_f(s, t)) \le -1$ if $f$ is the outer face
  \item $\rot_{\mathcal R}(\pi_f(s, t)) \le 5$
  \end{compactenum}
\end{lemma}
\begin{proof}
  We first consider the case where $f$ is the outer face
  (Figure~\ref{fig:strictly-directed-path}(a)), that is cases~(1)
  and~(2).  Due to the fact that $\pi_f(t, s)$ is strictly directed
  from $t$ to $s$ and the flexibility of every edge is positive, each
  edge on $\pi_f(t, s)$ has rotation at least~1.  Moreover, the
  rotations at vertices along the path $\pi_f(t, s)$ are at least~$-1$
  since $\pi_f(t, s)$ is simple as $G + st$ is biconnected.  Since the
  number of internal vertices on a path is one less than the number of
  edges this yields $\rot(\pi_f(t, s)) \ge 1$; see
  Figure~\ref{fig:strictly-directed-path}(b).  If $G$ consists of a
  single path this directly yields $\rot(\pi_f(s, t)) \le -1$ and thus
  concludes case~(2).  For case~(1) first assume that the degrees of
  $s$ and $t$ are not~1 (Figure~\ref{fig:strictly-directed-path}(b)),
  that is $\rot(s_f), \rot(t_f) \in \{-1, 0, 1\}$ holds.  Since $f$ is
  the outer face the equation $\rot(\pi_f(s, t)) + \rot(t_f) +
  \rot(\pi_f(t, s)) + \rot(s_f) = -4$ holds and directly implies the
  desired inequality $\rot(\pi_f(s, t)) \le -3$.  In the case that for
  example $t$ has degree~1 (and $\deg(s) > 0$), we have $\rot(t_f) =
  -2$ and $\rot(s_f) \in \{-1, 0, 1\}$, thus the considerations above
  only yield $\rot(\pi_f(s, t)) \le -2$.  However, in this case there
  necessarily exists a vertex $t'$ where the paths $\pi_f(s, t)$ and
  $\pi_f(t, s)$ split, as illustrated in
  Figure~\ref{fig:strictly-directed-path}(c).  More precisely,
  let~$t'$ be the first vertex on $\pi_f(s, t)$ that also belongs to
  $\pi_f(t, s)$.  Obviously, the degree of $t'$ is at least~3 and thus
  $\rot(t'_f)$ (with respect to the path $\pi_f(t, s)$) is at least~0.
  Hence we obtain the stronger inequality $\rot(\pi_f(t, s)) \ge 2$
  yielding the desired inequality $\rot(\pi_f(s, t)) \le -3$.  If $s$
  and $t$ both have degree~1 we cannot only find the vertex $t'$ but
  also the vertex $s'$ where the paths $\pi_f(s, t)$ and $\pi_f(t, s)$
  split.  Since $G + st$ is biconnected these two vertices are
  distinct and the estimation above works, finally yielding
  $\rot(\pi_f(s, t)) \le -3$.

  If $f$ is an internal face
  (Figure~\ref{fig:strictly-directed-path}(d)), that is case (3)
  applies, we start with the equation $\rot(\pi_f(s, t)) + \rot(t_f) +
  \rot(\pi_f(t, s)) + \rot(s_f) = 4$.  First we consider the case that
  neither $t$ nor~$s$ have degree~1.  Thus, $\rot(t_f), \rot(s_f) \in
  \{-1, 0, 1\}$.  With the same argument as above we obtain
  $\rot(\pi_f(t, s)) \ge 1$ and hence $\rot(\pi_f(s, t)) \le 5$; see
  Figure~\ref{fig:strictly-directed-path}(e).  Now assume that $t$ has
  degree~1 and~$s$ has larger degree.  Then $\rot(t_f) = -2$ holds and
  the above estimation does not work anymore.  Again, at some vertex
  $t'$ the paths $\pi_f(t, s)$ and $\pi_f(s, t)$ split as illustrated
  in Figure~\ref{fig:strictly-directed-path}(f).  Obviously, the
  degree of $t'$ needs to be greater than $2$ and thus $\rot(t'_f)$ is
  at least~0.  This yields $\rot(\pi_f(t, s)) \ge 2$ in the case that
  $\deg(t) = 1$, compensating $\rot(t_f) = -2$ (instead of $\rot(t_f)
  \ge -1$ in the other case).  To sum up, we obtain the desired
  inequality $\rot(\pi_f(s, t)) \le 5$.  The case $\deg(s) = \deg(t) =
  1$ works analogously.
\end{proof}

The \emph{flex graph} $G_{\mathcal R}^\times$ of $G$ with respect to a
valid orthogonal representation $\mathcal R$ is defined to be the dual
graph of $G$ such that the dual edge $e^\star$ is undirected if $e$ is
undirected, otherwise it is directed from the face right of $e$ to the
face left of $e$.  Figure~\ref{fig:reduce-rot}(a) shows an example
graph with an orthogonal drawing together with the corresponding flex
graph.  Assume we have a simple directed cycle $C$ in the flex graph.
Then \emph{bending} along this cycle yields a new valid orthogonal
representation $\mathcal R'$ which is defined as follows.  Let
$e^\star = (f_1, f_2)$ be an edge contained in $C$ dual to $e$.  Then
we decrease $\rot(e_{f_1})$ and increase $\rot(e_{f_2})$ by~1.  It can
be easily seen that the necessary properties for $\mathcal R'$ to be
an orthogonal representation are satisfied.  Obviously,
$\rot_{\mathcal R'}(e_{f_1}) = - \rot_{\mathcal R'}(e_{f_2})$ holds
and rotations at vertices did not change.  Moreover, the rotation
around a face $f$ does not change since $f$ is either not contained in
$C$ or it is contained in $C$, but then it has exactly one incoming
and exactly one outgoing edge.  Note that bending along a cycle in the
flex graph preserves the planar embedding of $G$ and for every vertex
the rotations in all incident faces.  The following lemma shows that a
high rotation along a path $\pi_f(s, t)$ for two vertices~$s$ and~$t$
sharing the face $f$ can be reduced by~1 using a directed cycle in the
flex graph.

\begin{lemma}
  \label{lem:reduce-rot}
  Let $G$ be a biconnected 4-planar graph with positive flexibility, a
  valid orthogonal representation~$\mathcal R$ and $s$ and $t$ on a
  common face $f$.  The flex graph $G_{\mathcal R}^\times$ contains a
  directed cycle $C$ such that $f \in C$, $s\in \lef(C)$ and $t\in
  \righ(C)$, if one of the following conditions holds.
  \begin{compactenum}[(1)]
  \item $\rot_{\mathcal R}(\pi_f(s, t)) \ge -2$, $f$ is the outer face
    and $\pi_f(s, t)$ is not strictly directed from $t$ to $s$
  \item $\rot_{\mathcal R}(\pi_f(s,t)) \ge 0$ and $f$ is the outer
    face
  \item $\rot_{\mathcal R}(\pi_f(s, t)) \ge 6$
  \end{compactenum}
\end{lemma}
\begin{proof}
  Figure~\ref{fig:reduce-rot}(b) shows the path $\pi_f(s, t)$ together
  with the desired cycle $C$.  Due to the duality of a cycle in the
  dual and a cut in the primal graph a directed cycle~$C$ in
  $G_{\mathcal R}^\times$ having $s$ and $t$ to the left and to the
  right of $C$, respectively, induces a directed cut in $G$ that is
  directed from $s$ to $t$ and vice versa.  Recall that directed
  cycles and cuts may also contain undirected edges.  Assume for
  contradiction that such a cycle $C$ does not exist.  
  \begin{claim}
    The graph $G$ contains a strictly directed path $\pi$ from $t$ to
    $s$.
  \end{claim}
  \noindent 
  Every cut $(S, T)$ with $T = V \setminus S$, $s \in S$ and~$t \in T$
  separating $s$ from $t$ must contain an edge that is directed from
  $T$ to $S$, otherwise this cut would correspond to a cycle $C$ in
  the flex graph that does not exist by assumption.  Let $T$ be the
  set of vertices in $G$ that can be reached by strictly directed
  paths from $t$.  If $T$ contains $s$ we found the path $\pi$
  strictly directed from $t$ to $s$.  Otherwise, $(S, T)$ with $S = V
  \setminus T$ is a cut separating $S$ from $T$ and there cannot be an
  edge that is directed from a vertex in $T$ to a vertex in $S$ which
  is a contradiction, and thus the path $\pi$ strictly directed from
  $t$ to $s$ exists, which concludes the proof of the claim.

  Let $G'$ be the subgraph of $G$ induced by the paths $\pi$ and
  $\pi_f(s, t)$ together with the orthogonal representation $\mathcal
  R'$ induced by $\mathcal R$.

  We first consider case~(1).  Let $f'$ be the outer face of the
  orthogonal representation $\mathcal R'$.  Obviously, $\pi_{f'}(s, t)
  = \pi_f(s, t)$ and $\pi = \pi_{f'}(t, s)$ holds, see
  Figure~\ref{fig:reduce-rot}(c).  Moreover, the graph $G' + st$ is
  biconnected and $G'$ does not consist of a single path since
  $\pi_{f'}(s, t)$ and $\pi_{f'}(t, s)$ are different due to the
  assumption that $\pi_f(s, t)$ is not strictly directed from $t$ to
  $s$.  Since $\pi_{f'}(t, s)$ is strictly directed from $t$ to $s$ we
  can use Lemma~\ref{lem:strictly-directed-path}(1) yielding
  $\rot_{\mathcal R'}(\pi_{f'}(s, t)) \le -3$ and thus $\rot_{\mathcal
    R}(\pi_f(s, t)) \le -3$, which is a contradiction.

  For case~(2) exactly the same argument holds except for the case
  where the strictly directed path $\pi$ is the path $\pi_f(s, t)$
  strictly directed from $t$ to $s$.  In this case we have to use
  Lemma~\ref{lem:strictly-directed-path}(2) instead of
  Lemma~\ref{lem:strictly-directed-path}(1) yielding $\rot_{\mathcal
    R}(\pi_f(s, t)) \le -1$, which is again a contradiction.

  In case~(3) the subgraph $G'$ of $G$ induced by the two paths $\pi$
  and $\pi_{f}(s, t)$ again contains~$s$ and~$t$ on a common face
  $f'$, which may be the outer or an inner face, see
  Figure~\ref{fig:reduce-rot}(c) and Figure~\ref{fig:reduce-rot}(d),
  respectively.  In both cases we obtain $\rot_{\mathcal R}(\pi_f(s,
  t)) \le 5$ due to Lemma~\ref{lem:strictly-directed-path}(3), which
  is a contradiction.
\end{proof}

\begin{figure}
  \centering
  \includegraphics{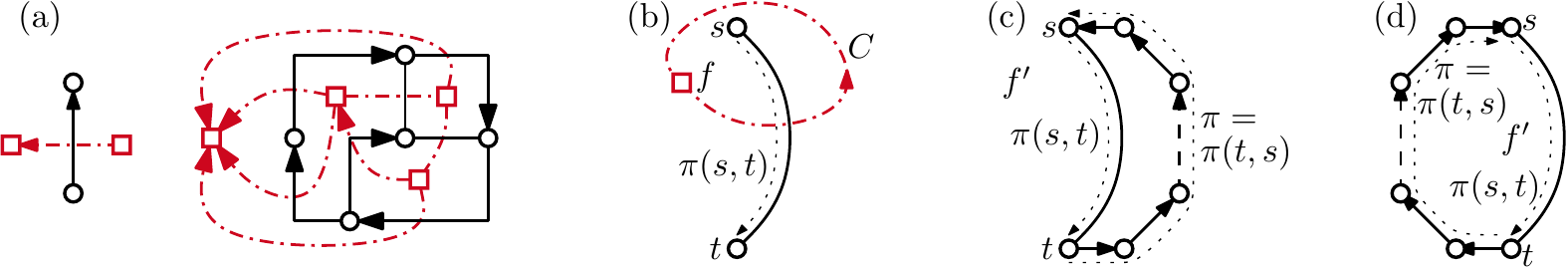}
  \caption{(a) An orthogonal representation and the corresponding flex
    graph where every edge has flexibility~1.  (b, c, d) Illustration
    of Lemma~\ref{lem:reduce-rot}.}
  \label{fig:reduce-rot}
\end{figure}

Lemma~\ref{lem:reduce-rot} directly yields the following corollary,
showing that graphs with positive flexibility behave very similar to
single edges with positive flexibility.

\begin{corollary}
  \label{cor:rotations-form-interval}
  Let $G$ be a graph with positive flexibility and vertices $s$ and
  $t$ such that $G + st$ is biconnected and 4-planar.  Let further
  $\mathcal R$ be a valid orthogonal representation with $s$ and $t$
  on the outer face $f$ such that $\rho = \rot_{\mathcal R}(\pi_f(s,
  t)) \ge 0$.  For every rotation $\rho' \in [-1, \rho]$ there exists
  a valid orthogonal representation $\mathcal R'$ with $\rot_{\mathcal
    R'}(\pi_f(s, t)) = \rho'$.
\end{corollary}
\begin{proof}
  For the case that $G$ itself is biconnected, the claim follows
  directly from Lemma~\ref{lem:reduce-rot}(2), since we can reduce the
  rotation along $\pi_f(s, t)$ stepwise by~1, starting with the
  orthogonal representation~$\mathcal R$, until we reach a rotation of
  $-1$.  For the case that $G$ itself is not biconnected we add the
  edge $\{s, t\}$ to the orthogonal representation $\mathcal R$ such
  that the path $\pi_f(s, t)$ does not change, that is $\pi_f(t, s)$
  consists of the new edge $\{s, t\}$.  Again
  Lemma~\ref{lem:reduce-rot}(2) can be used to reduce the rotation
  stepwise down to~$-1$.
\end{proof}

As edges with many bends imply the existence of paths with high
rotation, we can use Lemma~\ref{lem:reduce-rot} to successively reduce
the number of bends of every edge down to three, except for a single
edge on the outer face.  Since we only bend along cycles in the flex
graph, neither the embedding nor the angles around vertices are
changed.

\begin{theorem}
  \label{thm:three-bends-per-edge}
  Let $G$ be a biconnected 4-planar graph with positive flexibility,
  having a valid orthogonal representation.  Then $G$ has a valid
  orthogonal representation with the same planar embedding, the same
  angles around vertices and at most three bends per edge, except for
  at most one edge on the outer face with up to five bends.
\end{theorem}
\begin{proof}
  In the following we essentially pick an edge with more than three
  bends, reduce the number of bends by one and continue with the next
  edge.  After each of these reduction steps we set the flexibility of
  every edge down to $\max\{\rho, 1\}$, where $\rho$ is the number of
  bends it currently has.  This ensures that in the next step the
  number of bends of each edge either is decreased, remains as it is
  or is increased from zero to one.

  We start with an edge $e = \{s, t\}$ that is incident to two faces
  $f_1$ and $f_2$ and has more than three bends.  Due to the fact that
  we traverse inner faces in clockwise and the outer face in
  counter-clockwise direction, the edge $e$ forms in one of the two
  faces the path from $s$ to $t$ and in the other face the path from
  $t$ to $s$.  Assume without loss of generality that $\pi_{f_1}(t,
  s)$ and $\pi_{f_2}(s, t)$ are the paths on the boundary of $f_1$ and
  $f_2$, respectively, that consist of $e$.  Note that
  $\rot(\pi_{f_1}(t, s)) = - \rot(\pi_{f_2}(s, t))$ holds and we
  assume that $\rot(\pi_{f_1}(t, s))$ is not positive.  As $e$ was
  assumed to have more than three bends, the inequality
  $\rot(\pi_{f_1}(t, s)) \le -4$ holds.  We distinguish between the
  two cases that $f_1$ is an inner or the outer face.  We first
  consider the case that $f_1$ is an inner face;
  Figure~\ref{fig:three-bends}(a) illustrates this situation for the
  case where $e$ has four bends.  Then the rotations around the face
  $f_1$ sum up to~4.  As the rotations at the vertices $s$ and $t$ can
  be at most~1, we obtain $\rot(\pi_{f_1}(s, t)) \ge 6$.  Thus we can
  apply Lemma~\ref{lem:reduce-rot}(3) to reduce the rotation of
  $\pi_{f_1}(s, t)$ by bending along a cycle in the flex graph that
  contains $f_1$ and separates $s$ from $t$.  Obviously, this
  increases the rotation along $\pi_{f_1}(t, s)$ by~1 and thus reduces
  the number of bends of $e$ by~1.

  \begin{figure}
    \centering
    \includegraphics{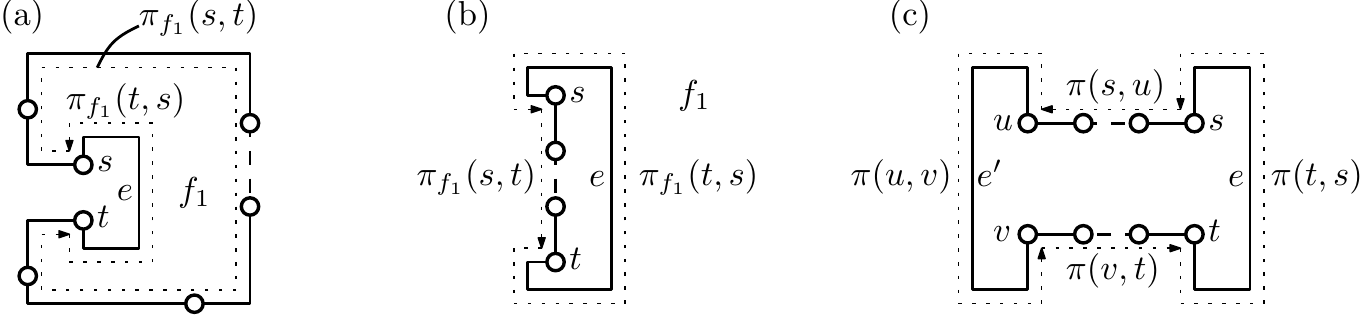}
    \caption{Reducing the number of bends on edges
      (Theorem~\ref{thm:three-bends-per-edge})}
    \label{fig:three-bends}
  \end{figure}

  For the case that $f_1$ is the outer face we first ignore the case
  where $e$ has four or five bends and show how to reduce the number
  of bends to five; Figure~\ref{fig:three-bends}(b) shows the case
  where $e$ has six bends.  Thus the inequality $\rot(\pi_{f_1}(t, s))
  \le -6$ holds.  As the rotations around the outer face $f_1$ sum up
  to $-4$ and the rotations at the vertices $s$ and $t$ are at most~1,
  the rotation along $\pi_{f_1}(s, t)$ must be at least~0.  Thus we
  can apply Lemma~\ref{lem:reduce-rot}(2) to reduce the rotation of
  $\pi_{f_1}(s, t)$ by~1, increasing the rotation along $\pi_{f_1}(t,
  s)$, and thus reducing the number of bends of $e$ by one.

  Finally, we obtain an orthogonal representation having at most three
  bends per edge except for some edges on the outer face with four or
  five bends having their negative rotation in the outer face.  If
  there is only one of these edges left we are done.  Otherwise let $e
  = \{s, t\}$ be one of the edges with $\rot(\pi_f(t, s)) \in \{-5,
  -4\}$, where $f$ is the outer face.  Then the inequality
  $\rot(\pi_f(s, t)) \ge -2$ holds by the same argument as before and
  we can apply Lemma~\ref{lem:reduce-rot}(1) to reduce the rotation,
  if we can ensure that~$\pi_f(s, t)$ is not strictly directed from
  $t$ to $s$.  To show that, we make use of the fact that~$\pi_f(s,
  t)$ contains an edge $e' = \{u, v\}$ with at least four bends due to
  the assumption that $e$ was not the only edge with more than three
  bends.  Assume without loss of generality that $u$ occurs before~$v$
  on $\pi_f(s, t)$, thus $\pi_f(s, t)$ splits into the three parts
  $\pi_f(s, u)$, $\pi_f(u, v)$ and $\pi_f(v, t)$.  Recall that
  $\rot(\pi_f(s, t)) \ge -2$ holds and thus $\rot(\pi_f(s, u)) +
  \rot(u) + \rot(\pi_f(u, v)) + \rot(v) + \rot(\pi_f(v, t)) \ge -2$.
  As the rotation at the vertices $u$ and $v$ is at most~1 and the
  rotation of $\pi_f(u, v)$ at most~$-4$ it follows that
  $\rot(\pi_f(s, u)) + \rot(\pi_f(v, t)) \ge 0$.
  Figure~\ref{fig:three-bends}(c) illustrates the situation for the
  case where~$e$ and $e'$ have four bends and $\rot(\pi_f(s, u)) =
  \rot(\pi_f(v, t)) = 0$.  Note that at least one of the two paths is
  not degenerate in the sense that $s \not= u$ or $v \not= t$,
  otherwise the total rotation around the outer face would be at most
  $-6$, which is a contradiction.  Assume without loss of generality
  that $\rot(\pi_f(s, u)) \ge 0$.  It follows that $\pi_f(s, u)$
  cannot be strictly directed from $u$ to $s$ and since~$\pi_f(s, u)$
  is a subpath of~$\pi_f(s, t)$ the path $\pi_f(s, t)$ cannot be
  strictly directed from $t$ to $s$.  This finally shows that we can
  use part~(1) of Lemma~\ref{lem:reduce-rot} implying that we can find
  a valid orthogonal representation such that at most a single edge
  with four or five bends remains, whereas all other edges have at
  most three bends.
\end{proof}

If we allow the embedding to be changed slightly, we obtain an even
stronger result.  Assume the edge $e$ lying on the outer face has more
than three bends.  If $e$ has five bends, we can reroute it in the
opposite direction around the rest of the graph, that is we can choose
the internal face incident to $e$ to be the new outer face.  In the
resulting drawing $e$ has obviously only three bends.  Thus the
following result directly follows from
Theorem~\ref{thm:three-bends-per-edge}.

\begin{corollary}
  \label{cor:three-four-bends-per-edge}
  Let $G$ be a biconnected 4-planar graph with positive flexibility
  having a valid orthogonal representation.  Then $G$ has a valid
  orthogonal representation with at most three bends per edge except
  for possibly a single edge on the outer face with four bends.
\end{corollary}

Note that Corollary~\ref{cor:three-four-bends-per-edge} is restricted
to biconnected graphs.  For general graphs it implies that each block
contains at most a single edge with up to four bends.
Figure~\ref{fig:lin-four-bend-edges} illustrates an instance of {\sc
  FlexDraw} with linearly many blocks and linearly many edges that are
required to have four bends, showing that
Corollary~\ref{cor:three-four-bends-per-edge} is tight.

\begin{figure}
  \centering
  \includegraphics[page=1]{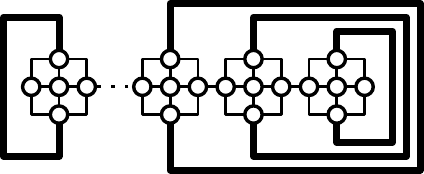}
  \caption{An instance of {\sc FlexDraw} requireing linearly many
    edges to have four bends.  Flexibilites are~1 except for the thick
    edges with flexibility~4.}
\label{fig:lin-four-bend-edges}
\end{figure}

Theorem~\ref{thm:three-bends-per-edge} implies that it is sufficient
to consider the flexibility of every edge to be at most~5, or in terms
of costs we want to optimize, it is sufficient to store the cost
function of an edge only in the interval $[0, 5]$.  However, there are
two reasons why we need a stronger result.  First, we want to compute
cost functions of split components and thus we have to limit the
number of ``bends'' they can have (see the next section for a precise
definition of bends for split components).  Second, as mentioned in
the introduction (see Figure~\ref{fig:not-convex-example}) the cost
function of a split component may already be non-convex on the
interval $[0, 5]$.  Fortunately, the second reason is not really a
problem since there may be at most a single edge with up to five
bends, all remaining edges have at most three bends and thus we only
need to consider their cost functions on the interval $[0, 3]$.  

In the following section we focus on dealing with the first problem
and strengthen the results so far presented by extending the
limitation on the number of bends to split components.  Note that a
split pair inside an inner face of $G$ with a split component $H$
having a rotation less than $-3$ on its outer face implies a rotation
of at least~6 in some inner face of $G$.  Thus, we can again apply
Lemma~\ref{lem:reduce-rot}(3) to reduce the rotation showing that
split components and single edges can be handled similarly.  However,
by reducing the rotation for one split component, we cannot avoid that
the rotation of some other split component is increased.  For single
edges we did that by reducing the flexibility to the current number of
bends.  In the following section we extend this technique by defining
a flexibility not only for edges but also for split components.  We
essentially show that all results we presented so far still apply, if
we allow this kind of extended flexibilities.

\section{Flexibility of Split Components and Nice Drawings}
\label{sec:flex-split-comp}

Let $G$ be a biconnected 4-planar graph with SPQR-tree $\mathcal T$
and let $\mathcal T$ be rooted at some node $\tau$.  Recall that we do
not require $\tau$ to be a Q-node.  Let $\mu$ be a node of $\mathcal
T$ that is not the root $\tau$.  Then~$\mu$ has a unique parent and
$\skel(\mu)$ contains a unique virtual edge $\eps = \{s, t\}$ that is
associated with this parent.  We call the split-pair $\{s, t\}$ a
\emph{principal split pair} and the pertinent graph $\pert(\mu)$ with
respect to the chosen root a \emph{principal split component}.  The
vertices $s$ and $t$ are the \emph{poles} of this split component.
Note that a single edge is also a principal split component except for
the case that its Q-node is chosen to be the root.  A planar embedding
of $G$ is represented by $\mathcal T$ with the root $\tau$ if the
embedding of each skeleton has the edge associated with the parent on
the outer face.

Let $\mathcal R$ be a valid orthogonal representation of $G$ such that
the planar embedding of $\mathcal R$ is represented by $\mathcal T$
rooted at $\tau$.  Consider a principal split component $H$ with
respect to the split pair $\{s, t\}$ and let $\mathcal S$ be the
orthogonal representation of $H$ induced by $\mathcal R$.  Note that
the poles $s$ and~$t$ are on the outer face $f$ of $\mathcal S$.  We
define $\max\{|\rot_{\mathcal S}(\pi_f(s, t))|, |\rot_{\mathcal
  S}(\pi_f(t, s))|\}$ to be the \emph{number of bends} of the split
component $H$.  Note that this is a straightforward extension of the
term \emph{bends} as it is used for edges.  With this terminology we
can assign a \emph{flexibility} $\flex(H)$ to a principal split
component $H$ and we define the orthogonal representation $\mathcal R$
of $G$ to be \emph{valid} if and only if $H$ has at most $\flex(H)$
bends.  We say that the graph $G$ has \emph{positive flexibility} if
the flexibility of every principal split component is at least~1,
which is straightforward extension of the original notion.

We define a valid orthogonal representation of $G$ to be \emph{nice}
if it is tight and if there is a root $\tau$ of the SPQR-tree such
that every principal split component has at most three bends and the
edge corresponding to~$\tau$ in the case that $\tau$ is a Q-node has
at most five bends.  The main result of this section will be the
following theorem, which directly extends
Theorem~\ref{thm:three-bends-per-edge}.

\newcommand{\thmThreeBendsSplitCompText}{Every biconnected 4-planar
  graph with positive flexibility having a valid orthogonal
  representation has an orthogonal representation with the same planar
  embedding and the same angles around vertices that is nice with
  respect to at least one node chosen as root of its SPQR-tree.}

\begin{theorem}
  \label{thm:three-bends-per-princ-spl-comp}
  \thmThreeBendsSplitCompText
\end{theorem}

Before we prove Theorem~\ref{thm:three-bends-per-princ-spl-comp} we
need to make some additional considerations.  In particular we need to
extend the flex-graph such that it takes the flexibilities of
principal split components into account.  The extended version of the
flex graph can then be used to obtain a result similar to
Lemma~\ref{lem:reduce-rot}, which was the main tool to proof
Theorem~\ref{thm:three-bends-per-edge}.  Another difficulty is that it
depends on the chosen root which split components are principal split
components.  For the moment we avoid this problem by choosing an
arbitrary Q-node to be the root of the SPQR-tree $\mathcal T$.  Thus
we only have to care about the flexibilities of the principal split
components with respect to the chosen root.  One might hope that the
considerations we make for the flex-graph in the case of a fixed root
still work, if we consider the principal split components with respect
to all possible roots at the same time.  However, this fails as we
will see later, making it necessary to consider internal vertices as
the root.

Assume that the SPQR-tree $\mathcal T$ of $G$ is rooted at the Q-node
corresponding to an arbitrary chosen edge.  Let $H$ be a principal
split component with respect to the chosen root with the poles~$s$
and~$t$.  In the embedding of $G$ the outer face $f$ of $H$ splits
into two faces $f_1$ and $f_2$, where the path $\pi_f(s, t)$ is
assumed to lie in $f_1$ and $\pi_f(t, s)$ is assumed to lie in $f_2$,
that is $\pi_{f_1}(s, t) = \pi_f(s, t)$ and $\pi_{f_2}(t, s) =
\pi_f(t, s)$.  We augment $G$ by inserting the edge $\{s, t\}$ twice,
embedding one of them in $f_1$ and the other in $f_2$.  We denote the
edge $\{s, t\}$ inserted into the face $f_1$ by $e_H(s, t)$ and the
edge inserted into $f_2$ by $e_H(t, s)$.
Figure~\ref{fig:safety-edges} illustrates this process and shows how
the dual graph of $G$ changes.  We call the new edges $e_H(s, t)$ and
$e_H(t, s)$ \emph{safety edges} and define the \emph{extended flex
  graph}~$G^\times$ as before, ignoring that some edges have a special
meaning.  To simplify notation we often use the term flex graph,
although we refer to the extended flex graph.  Note that every cycle
in the flex graph that separates $s$ from $t$ and thus crosses $\pi(s,
t)$ and $\pi(t, s)$ needs to also cross the safety edges $e_H(s, t)$
and $e_H(t, s)$.  Thus we can use the safety edges to ensure that the
flex graph respects the flexibility of $H$ by orienting them if
necessary.  More precisely, we orient the safety edge $e_H(s, t)$ from
$t$ to $s$ if $\rot(\pi(s, t)) = -\flex(H)$ and similarly $e_H(t, s)$
from $s$ to~$t$ if $\rot(\pi(t, s)) = -\flex(H)$.  This ensures that
the rotations along $\pi(s, t)$ and $\pi(t, s)$ cannot be reduced
below $-\flex(H)$ by bending along a cycle in the flex graph.
Moreover, $\rot(\pi(s, t))$ cannot be increased above $\flex(H)$ as
otherwise $\rot(\pi(t, s))$ has to be below $-\flex(H)$ and vice
versa.  To sum up, we insert the safety edges next to the principal
split component $H$ and orient them if necessary to ensure that
bending along a cycle in the flex graph respects not only the
flexibilities of single edges but also the flexibility of the
principal split component $H$.

\begin{figure}
  \centering
  \includegraphics[page=1]{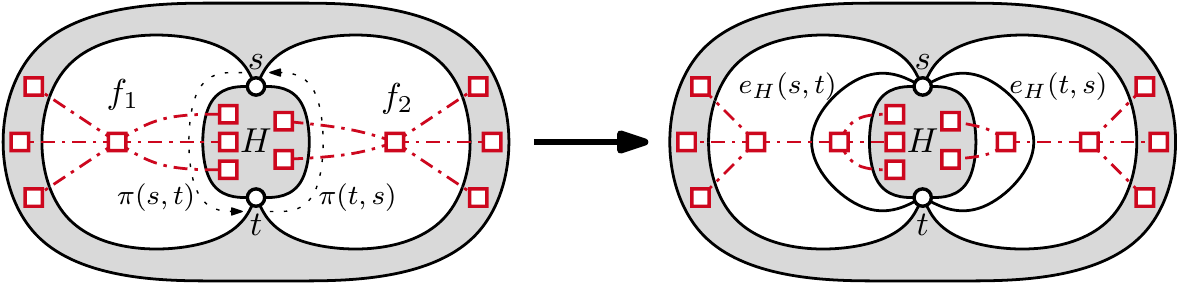}
  \caption{Augmentation of $G$ by the safety edges $e_H(s, t)$ and
    $e_H(t, s)$.}
  \label{fig:safety-edges}
\end{figure}

Since adding the safety edges for the graph $H$ is just a technique to
respect the flexibility of~$H$ by bending along a cycle in the flex
graph, we do not draw them.  Note that the augmented graph does not
have maximum degree~4 anymore but this is not a problem since we do
not draw the safety edges.  However, we formally assign an orthogonal
representation to the safety edges by essentially giving them the
shape of the paths they ``supervise''.  More precisely, the edges
$e_H(s, t)$ and~$e_H(t, s)$ have the same rotations as the paths
$\pi(s, t)$ and $\pi(t, s)$ on the outer face of $H$, respectively.
Moreover, the angles at the vertices $s$ and $t$ are also assumed to
be the same as for these two paths.

As we do not only want to respect the flexibility of a single split
component, we add the safety edges for each of the principal split
components at the same time.  Note that the augmented graph remains
planar as we only add the safety edges for the principal split
components with respect to a single root.  It follows directly that
the considerations above still work, which would fail if the augmented
graph was non-planar.  This is the reason why we cannot consider the
principal split components with respect to all roots at the same time.
The following lemma directly extends Lemma~\ref{lem:reduce-rot} to the
case where the extended flex graph is considered.

\begin{lemma}
  \label{lem:reduce-rot-extended}
  Let $G$ be a biconnected 4-planar graph with positive flexibility, a
  valid orthogonal representation~$\mathcal R$ and $s$ and $t$ on a
  common face $f$.  The extended flex graph $G_{\mathcal R}^\times$
  contains a directed cycle $C$ such that $f \in C$, $s\in \lef(C)$
  and $t\in \righ(C)$, if one of the following conditions holds.
  \begin{compactenum}[(1)]
  \item $\rot_{\mathcal R}(\pi_f(s, t)) \ge -2$, $f$ is the outer face
    and $\pi_f(s, t)$ is not strictly directed from $t$ to $s$
  \item $\rot_{\mathcal R}(\pi_f(s,t)) \ge 0$ and $f$ is the outer
    face
  \item $\rot_{\mathcal R}(\pi_f(s, t)) \ge 6$
  \end{compactenum}
\end{lemma}
\begin{proof}
  As in the proof of Lemma~\ref{lem:reduce-rot} we assume for
  contradiction that the cycle $C$ does not exists, yielding a
  strictly directed path from $t$ to $s$ in $G$.  This directly yields
  the claim, if we can apply Lemma~\ref{lem:strictly-directed-path} as
  before.  The only difference to the situation before is that the
  directed path from $t$ to~$s$ may contain some of the safety edges.
  However, by definition a safety edge $e_H(u, v)$ is directed from
  $v$ to $u$ if and only if $\rot(\pi(u, v)) = -\flex(H)$.  As
  $\flex(H)$ is positive $\rot(\pi(u, v))$ has to be negative and thus
  the rotation along $e_H(u, v)$ when traversing it from $v$ to $u$ is
  at least~1.  Thus, it does not make a difference whether the
  directed path from $t$ to $s$ consists of normal edges or may
  contain safety edges.  Hence, Lemma~\ref{lem:strictly-directed-path}
  extends to the augmented graph containing the safety edges, which
  concludes the proof.
\end{proof}

Now we are ready to prove
Theorem~\ref{thm:three-bends-per-princ-spl-comp}.  To improve
readability we state it again.

\begin{rethm}{thm:three-bends-per-princ-spl-comp}
  \thmThreeBendsSplitCompText
\end{rethm}
\begin{proof}
  Let $\mathcal R$ be a valid orthogonal representation of $G$.  We
  assume without loss of generality that~$\mathcal R$ is tight.  Since
  the operations we apply to $\mathcal R$ in the following do not
  affect the angles around vertices, the resulting orthogonal
  representation is also tight.  Thus it remains to enforce the more
  interesting condition for orthogonal representations to be nice,
  that is reduce the number of bends of principal split components
  down to three.  As mentioned before, the SPQR-tree $\mathcal T$ of
  $G$ is initially rooted at an arbitrary Q-node.  Let $\eref$ be the
  corresponding edge.  As in the proof of
  Theorem~\ref{thm:three-bends-per-edge} we start with an arbitrary
  principal split component $H$ with more than three bends.  Then one
  of the two paths in the outer face of $H$ has rotation less
  than~$-3$ and we have the same situation as for a single edge, that
  is we can apply Lemma~\ref{lem:reduce-rot-extended} to reduce the
  rotation of the opposite site and thus reduce the number of bends of
  $H$ by one.  Afterwards, we can set the flexibility of $H$ down to
  the new number of bends ensuring that it is not increased later on.
  However, this only works if the negative rotation of the split
  component $H$ lies in an inner face of $G$.  On the outer face we
  can only increase to a rotation of $-5$ yielding an orthogonal
  representation such that every principal split component has at most
  three bends, or maybe four or five bends, if it has its negative
  rotation in the outer face.  Note that this is essentially the same
  situation we also had in the proof of
  Theorem~\ref{thm:three-bends-per-edge}.  In the following we show
  similarly that the number of bends can be reduced further, until
  either a unique innermost principal split component (where innermost
  means minimal with respect to inclusion) or the reference edge
  $\eref$ may have more than three bends.

  First assume that $\eref$ has more than three, that is four or five,
  bends and that there is a principal split component $H$ with more
  than three bends having its negative rotation on the outer face.
  Let~$\{s, t\}$ be the corresponding split pair and let without loss
  of generality $\pi_f(t, s)$ be the path along $H$ with rotation less
  than~$-3$ where $f$ is the outer face.  Then the path $\pi_f(s, t)$
  contains the edge~$\eref = \{u, v\}$, otherwise $H$ would not be a
  principal split component.  Moreover, $\rot(\pi_f(t, s)) \le -4$
  implies that $\rot(\pi_f(s, t)) \ge -2$ holds.  As in the proof of
  Theorem~\ref{thm:three-bends-per-edge} (compare with
  Figure~\ref{fig:three-bends}(c)) the path $\pi_f(s, t)$ splits into
  the paths $\pi_f(s, u)$, $\pi_f(u, v)$ and $\pi_f(v, t)$.  Since
  $\pi_f(u, v)$ consists of the single edge $\eref$ with more than
  three bends $\rot(\pi_f(u, v)) \le -4$ holds, implying that the
  rotation along $\pi_f(s, u)$ or $\pi_f(v, t)$ is greater or equal
  to~0.  This shows that $\pi_f(s, t)$ cannot be strictly directed
  from $t$ to $s$ and thus we can apply
  Lemma~\ref{lem:reduce-rot-extended}(1) to reduce the number of bends
  $H$ has.  Finally, there is no principal split component with more
  than three bends left and the reference edge $\eref$ has at most
  five bends, which concludes this case.

  In the second case, $\eref$ has at most three bends.  We show that
  if there is more than one principal split component with more than
  three bends, then they hierarchically contain each other.  Assume
  that the number of bends of no principal split component that has
  more than three bends can be reduced further.  Assume further there
  are two principal split components $H_1$ and $H_2$ with respect to
  the split pairs $\{s_1, t_1\}$ and $\{s_2, t_2\}$ that do not
  contain each other, that is without loss of generality the vertices
  $t_1, s_1, t_2$ and $s_2$ occur in this order around the outer face
  $f$ when traversing it in counter-clockwise direction and
  $\pi_f(t_1, s_1)$ and $\pi_f(t_2, s_2)$ belong to $H_1$ and $H_2$
  respectively.  Analogous to the case where $\eref$ has more than
  three bends we can show that Lemma~\ref{lem:reduce-rot-extended}(1)
  can be applied to reduce the number of bends of $H_1$, which is a
  contradiction.  Thus, either $H_1$ is contained in $H_2$ or the
  other way round.  This shows that there is a unique principal split
  component~$H$ that is minimal with respect to inclusion having more
  than three bends.  Due to the inclusion property, all nodes in the
  SPQR-tree corresponding to the principal split components with more
  than three bends lie on the path between the current root and the
  node corresponding to $H$.  We denote the node corresponding to $H$
  by $\tau$ and choose $\tau$ to be the new root of the
  SPQR-tree~$\mathcal T$.  Since the principal split components depend
  on the root chosen for $\mathcal T$ some split components may no
  longer be principal and some may become principal due to rerooting.
  Our claim is that all principal split components with more than
  three bends are no longer principal after rerooting and furthermore
  that all split components becoming principal can be enforced to have
  at most three bends.

  \begin{figure}
    \centering
    \includegraphics{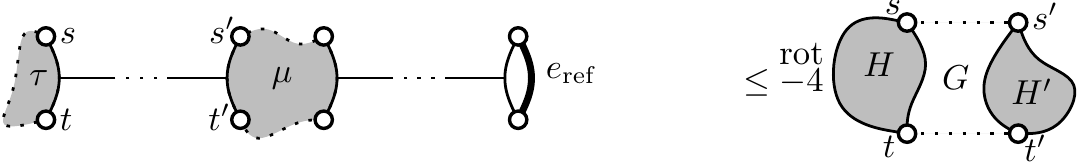}
    \caption{The path between the new and the old root in the
      SPQR-tree containing $\mu$ (left).  The whole graph $G$
      containing the principal split component $H'$ corresponding to
      $\mu$ with respect to the new root and the principal split
      component $H$ of the new root with respect to the old root
      (right).}
    \label{fig:rerooting}
  \end{figure}

  First note that the principal split component corresponding to a
  node $\mu$ in the SPQR-tree changes if and only if $\mu$ lies on the
  path between the old and the new root, that is between $\tau$ and
  the Q-node corresponding to $\eref$.  Since all principal split
  components (with respect to the old root) that have more than three
  bends also lie on this path, all these split components are no
  longer principal (with respect to the new root).  It remains to deal
  with the new principal split components corresponding to the nodes
  on this path.  Note that the new root $\tau$ itself has no principal
  split component associated with it.  Let $\mu \not= \tau$ be a node
  on the path between the new and the old root and let $H'$ be the new
  principal split component corresponding to $\mu$ with the poles $s'$
  and $t'$.  Recall that $H$ is the former principal split component
  corresponding to the new root $\tau$ with the poles $s$ and $t$.
  Note that $H$ of course is still a split component, although it is
  not principal anymore.  Figure~\ref{fig:rerooting} illustrates this
  situation.  Now assume that $H'$ has more than three bends.  Then
  there are two possibilities, either it has its negative rotation on
  the outer face or in some inner face. If only the latter case arises
  we can easily reduce the number of bends down to three as we did
  before.  In the remaining part of the proof we show that the former
  case cannot arise due to the assumption that the number of bends of
  $H$ cannot be reduced anymore.  Assume~$H'$ has its negative
  rotation in the outer face $f$, that is without loss of generality
  the path $\pi_f(t, s)$ belongs to $H'$ and has rotation at
  most~$-4$.  Thus we have again the situation that the two split
  components $H'$ and $H$ both have a rotation of at most~$-4$ in the
  outer face.  Moreover, these two split components do not contain or
  overlap each other since $s$ and $t$ are not contained in $H'$ as
  $\tau$ is the new root and $H$ does not contain $s'$ or $t'$ since
  $\mu$ is an ancestor of $\tau$ with respect to the old root.  Thus
  we could have reduced the number of bends of $H$ before we changed
  the root, which is a contradiction to the assumption we made that
  the number of bends of principal split components with more than
  three bends cannot be reduced anymore.  Hence, all new principal
  split components either have at most three bends or they have their
  negative rotation in some inner face.  Finally, we obtain a valid
  orthogonal representation with at most three bends per principal
  split component with respect to~$\tau$.
\end{proof}

\section{Optimal Drawings with Fixed Planar Embedding}
\label{sec:optim-draw-with}

All results from the previous sections deal with the case where we are
only interested in the decision problem of whether a given graph has a
valid drawing or not.  More precisely, we always assumed to have a
valid orthogonal representation of an instance of {\sc FlexDraw} and
showed that this implies that there exists another valid orthogonal
representation with certain properties.  In this section, we consider
convex instances of the optimization problem {\sc OptimalFlexDraw}.
The following generic theorem shows that the results for {\sc
  FlexDraw} that we presented so far can be extended to {\sc
  OptimalFlexDraw}.

\begin{theorem}
  \label{thm:valid-and-optimal}
  If the existence of a valid orthogonal representation of an instance
  of {\sc FlexDraw} with positive flexibility implies the existence of
  a valid orthogonal representation with property $P$, then every
  convex instance of {\sc OptimalFlexDraw} has an optimal drawing with
  property $P$.
\end{theorem}
\begin{proof}
  Let $G$ be a convex instance of {\sc OptimalFlexDraw}.  Let further
  $\mathcal R$ be an optimal orthogonal representation.  We can
  reinterpret $G$ as an instance of {\sc FlexDraw} with positive
  flexibility by setting the flexibility of an edge with $\rho$ bends
  in $\mathcal R$ to $\max\{\rho, 1\}$.  Then $\mathcal R$ is obviously
  a valid orthogonal representation of $G$ with respect to these
  flexibilities.  Thus there exists another valid orthogonal
  representation $\mathcal R'$ having property $P$.  It remains to
  show that $\cost(\mathcal R') \le \cost(\mathcal R)$ holds when going back
  to the optimization problem {\sc OptimalFlexDraw}.  However, this is
  clear for the following reason.  Every edge $e$ has as most as many
  bends in $\mathcal R'$ as in $\mathcal R$ except for the case
  where~$e$ has one bend in $\mathcal R'$ and zero bends in $\mathcal
  R$.  In the former case the monotony of $\cost_e(\cdot)$ implies
  that the cost did not increase.  In the latter case $e$ causes the
  same amount of cost in $\mathcal R$ as in $\mathcal R'$ since
  $\cost_e(0) = \cost_e(1) = b_e$ holds for convex instances of {\sc
    OptimalFlexDraw}.  Note that this proof still works, if the cost
  functions are only monotone but not convex.
\end{proof}

It follows that every convex 4-planar graph has an optimal drawing
that is nice since Theorem~\ref{thm:valid-and-optimal} shows that
Theorem~\ref{thm:three-bends-per-princ-spl-comp} can be applied.
Thus, it is sufficient to consider only nice drawings when searching
for an optimal solution, as there exists a nice optimal solution.
This is a fact that we crucially exploit in the next section since
although the cost function of a principal split component may be
non-convex, we can show that it is convex in the interval that is of
interest when only considering nice drawings.

\section{Optimal Drawings with Variable Planar Embedding}
\label{sec:opt-draw-var-emb}

All results we presented so far were based on a fixed planar embedding
of the input graph $G$.  In this section we present an algorithm that
computes an optimal drawing of $G$ in polynomial time, optimizing over
all planar embeddings of $G$.  Our algorithm crucially relies on the
existence of a nice drawing among all optimal drawings of $G$.  For
biconnected graphs (Section~\ref{sec:biconnected-graphs}) we present a
dynamic program that computes the cost function of all principal split
components bottom-up in the SPQR-tree with respect to a chosen root.
To compute the optimal drawing among all drawings that are nice with
respect to the chosen root, it remains to consider the embeddings of
the root itself.  If we choose every node to be the root once, this
directly yields an optimal drawing of $G$ taking all planar embeddings
into account.  In Section~\ref{sec:connected-graphs} we extend our
results to connected graphs that are not necessarily biconnected.  To
this end we first modify the algorithm for biconnected graphs such
that it can compute an optimal drawing with the additional requirement
that a specific vertex lies on the outer face.  Then we can use the
BC-tree to solve {\sc OptimalFlexDraw} for connected graphs.  We use
the computation of a minimum-cost flow in a network of size $n$ as a
subroutine and denote the consumed running time by~$T_{\flow}(n)$.  In
Section~\ref{sec:computing-flow} we consider which running time we
actually need.

\subsection{Biconnected Graphs}
\label{sec:biconnected-graphs}

In this section we always assume $G$ to be a biconnected 4-planar
graph forming a convex instance of {\sc OptimalFlexDraw}.  Let
$\mathcal T$ be the SPQR-tree of $G$.  As defined before, an
orthogonal representation is optimal if it has the smallest possible
cost.  We call an orthogonal representation \emph{$\tau$-optimal} if
it has the smallest possible cost among all orthogonal representation
that are nice with respect to the root $\tau$.  We say that it
is~\emph{$(\tau, \mathcal E)$-optimal} if it causes the smallest
possible amount of cost among all orthogonal representations that are
nice with respect to $\tau$ and induce the planar embedding $\mathcal
E$ on $\skel(\tau)$.  In this section we concentrate on finding a
$(\tau, \mathcal E)$-optimal orthogonal representation with respect to
a root $\tau$ and a given planar embedding $\mathcal E$ of
$\skel(\tau)$.  Then a $\tau$-optimal representation can be computed
by choosing every possible embedding of $\skel(\tau)$.  An optimal
solution can then be computed by choosing every node in $\mathcal T$
to be the root once.

\begin{figure}
  \centering
  \includegraphics[page=1]{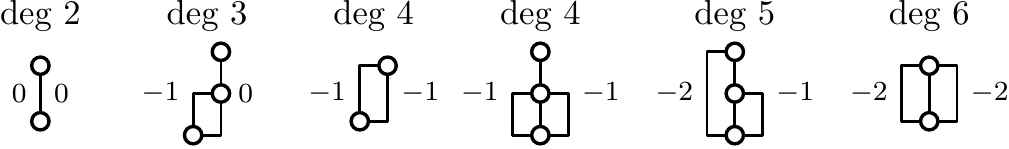}
  \caption{Split components with as few bends as possible.}
  \label{fig:degree-and-s-node-root}
\end{figure}

In Section~\ref{sec:flex-split-comp} we extended the terms ``bends''
and ``flexibility'', which were originally defined for single edges,
to arbitrary principal split components with respect to the chosen
root.  We start out by making precise what we mean with the cost
function $\cost_H(\cdot)$ of a principal split component~$H$ with
poles $s$ and $t$.  Recall that the number of bends of $H$ with
respect to an orthogonal representation $\mathcal S$ with $s$ and $t$
on the outer face $f$ is defined to be $\max\{|\rot_{\mathcal
  S}(\pi_f(s, t))|, |\rot_{\mathcal S}(\pi_f(t, s))|\}$.  Assume
$\mathcal S$ is the nice orthogonal representation of $H$ that has the
smallest possible cost among all nice orthogonal representations with
$\rho$ bends.  Then we essentially define $\cost_H(\rho)$ to be the
cost of $\mathcal S$.  However, with this definition the cost function
of $H$ is not defined for all $\rho \in \mathbb N_0$ since $H$ does
not have an orthogonal representation with zero bends at all, if
$\deg(s) > 1$ or $\deg(t) > 1$, as at least one of the paths $\pi_f(s,
t)$ and $\pi_f(t, s)$ has negative rotation in this case.  More
precisely, if $\deg(s) + \deg(t) > 2$, then $H$ has at least one bend,
and if $\deg(s) + \deg(t) > 4$, then $H$ has at least two bends.
Figure~\ref{fig:degree-and-s-node-root} shows for each combination of
degrees a small example with the smallest possible number of bends.
In these two cases we formally set $\cost_H(0) = \cost_H(1)$ and
$\cost_H(0), \cost_H(1) = \cost_H(2)$, respectively.  Thus, we only
need to compute the cost functions for at least $\lceil(\deg(s) +
\deg(t) - 2)/2\rceil$ bends.  We denote this lower bound by $\ell_H =
\lceil(\deg(s) + \deg(t) - 2)/2\rceil$.  Hence, it remains to compute
the cost function $\cost_H(\rho)$ for $\rho \in [\ell_H, 3]$.  For
more than three bends we formally set the cost to $\infty$.  Note that
the definition of the cost function only considers nice orthogonal
representations (including that they are tight).  As a result of this
restriction the cost for an orthogonal representation with $\rho$
bends might be less than $\cost_H(\rho)$.  However, due to
Theorem~\ref{thm:three-bends-per-princ-spl-comp} in combination with
Theorem~\ref{thm:valid-and-optimal} we know that optimizing over nice
orthogonal representations is sufficient to find an optimal solution.

As for single edges, we define the \emph{base cost} $b_H$ of the
principal split component $H$ to be $\cost_H(0)$.  We will see that
the cost function $\cost_H(\cdot)$ is monotone and even convex in the
interval $[0, 3]$ (except for a special case) and thus the base cost
is the smallest possible amount of cost that has to be payed for every
orthogonal drawing of $H$.  The only exception is the case where
$\deg(s) = \deg(t) = 3$.  In this case $H$ has at least two bends and
thus the cost function $\cost_H(\cdot)$ needs to be considered only on
the interval $[2, 3]$.  However, it may happen that $\cost_H(2) >
\cost_H(3)$ holds in this case.  Then we set the base cost $b_H$ to
$\cost_H(3)$ such that the base cost $b_H$ is really the smallest
possible amount of cost that need to be payed for every orthogonal
representation of $H$.  We obtain the following theorem.

\begin{theorem}
  \label{thm:convex-cost-function}
  If the poles of a principal split component do not both have
  degree~3, then its cost function is convex on the interval $[0, 3]$.
\end{theorem}

Before showing Theorem~\ref{thm:convex-cost-function} we just assume
that it holds and moreover we assume that the cost function of every
principal split component is already computed.  We first show how
these cost functions can then be used to compute an optimal drawing.
To this end, we define a flow network on the skeleton of the root
$\tau$ of the SPQR-tree, similar to Tamassias flow
network~\cite{t-eggmb-87}.  The cost functions computed for the
children of $\tau$ will be used as cost functions on arcs in the flow
network.  As we can only solve flow networks with convex costs we
somehow have to deal with potentially non-convex cost functions for
the case that both endvertices of a virtual edge have degree~3 in its
expansion graph.  Our strategy is to simply ignore these subgraphs by
contracting them into single vertices.  Note that the resulting
vertices have degree~2 since the poles of graphs with non-convex cost
functions have degree~3.  The process of replacing the single vertex
in the resulting drawing by the contracted component is illustrated in
Figure~\ref{fig:contraction}.  The following lemma justifies this
strategy.

\begin{figure}
  \centering
  \includegraphics[page=1]{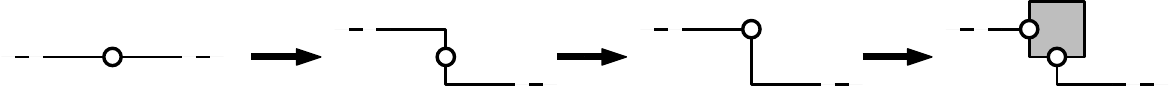}
  \caption{A single vertex can be replaced by a split component with
    three bends.}
  \label{fig:contraction}
\end{figure}

\begin{lemma}
  \label{lem:contract-deg-33}
  Let $G$ be a biconnected convex instance of {\sc OptimalFlexDraw}
  with $\tau$-optimal orthogonal representation $\mathcal R$ and let
  $H$ be a principal split component with non-convex cost function and
  base cost $b_H$.  Let further $G'$ be the graph obtained from $G$ by
  contracting $H$ into a single vertex and let $\mathcal R'$ be a
  $\tau$-optimal orthogonal representation of~$G'$.  Then
  $\cost(\mathcal R) = \cost(\mathcal R') + b_H$ holds.
\end{lemma}
\begin{proof}
  Assume we have a $\tau$-optimal orthogonal representation $\mathcal
  R$ of $G$ inducing the orthogonal representation $\mathcal S$ on
  $H$.  As $H$ has either two or three bends we can simply contract it
  yielding an orthogonal representation $\mathcal R'$ of $G$ with
  $\cost(\mathcal R') = \cost(\mathcal R) - \cost(\mathcal S) \le
  \cost(\mathcal R) - b_H$.  The opposite direction is more
  complicated.  Assume we have an orthogonal representation $\mathcal
  R'$ of $G'$, then we want to construct an orthogonal representation
  $\mathcal R$ of $G$ with $\cost(\mathcal R) = \cost(\mathcal R') +
  b_H$.  Let~$\mathcal S$ be an orthogonal representation of $H$
  causing only $b_H$ cost.  Since $\cost_H(\cdot)$ was assumed to be
  non-convex,~$\mathcal S$ needs to have three bends.  It is easy to
  see that $\mathcal R'$ and $\mathcal S$ (or $\mathcal S'$ obtained
  from $\mathcal S$ by mirroring the drawing) can be combined to an
  orthogonal representation of $G$ if the two edges incident to the
  vertex $v$ in $G'$ corresponding to $H$ have an angle of $90^\circ$
  between them.  However, this can always be ensured without
  increasing the costs of $\mathcal R'$.  Let $e_1$ and $e_2$ be the
  edges incident to $v$ and assume they have an angle of $180^\circ$
  between them in both faces incident to $v$.  If neither~$e_1$ nor
  $e_2$ has a bend, the flex graph contains the cycle around $v$ due
  to the fact that~$e_1$ and $e_2$ have positive flexibilities.
  Bending along this cycles introduces a bend to each of the edges,
  thus we can assume without loss of generality that $e_1$ has a bend
  in $\mathcal R'$.  Moving $v$ along the edge $e_1$ until it reaches
  this bend decreases the number of bends on $e_1$ by one and ensures
  that~$v$ has an angle of $90^\circ$ in one of its incident faces.
  Thus we can replace $v$ by the split component $H$ with orthogonal
  representation~$\mathcal S$ having cost $b_H$ yielding an orthogonal
  representation $\mathcal R$ of $G$ with $\cost(\mathcal R) =
  \cost(\mathcal R') + b_H$.
\end{proof}

When computing a $(\tau, \mathcal E)$-optimal orthogonal representation
of $G$ we make use of Lemma~\ref{lem:contract-deg-33} in the following
way.  If the expansion graph $H$ corresponding to a virtual edge
$\eps$ in $\skel(\tau)$ has a non-convex cost function, we simply
contract this virtual edge in $\skel(\tau)$.  Note that this is
equivalent to contracting $H$ in $G$.  We can then make use of the
fact that all remaining expansion graphs have convex cost functions to
compute a $(\tau, \mathcal E)$-optimal orthogonal representation of the
resulting graph yielding a $(\tau, \mathcal E)$-optimal orthogonal
representation of the original graph $G$ since the contracted
expansion graphs can be inserted due to
Lemma~\ref{lem:contract-deg-33}.  Note that expansion graphs with non
convex cost functions can only appear if the root is a Q- or an
S-node.  In the skeletons of P- and R-nodes every vertex has degree at
least three, thus the poles of an expansion graph cannot have degree~3
since $G$ has maximum degree~4.

\begin{figure}
  \centering
  \includegraphics[page=1]{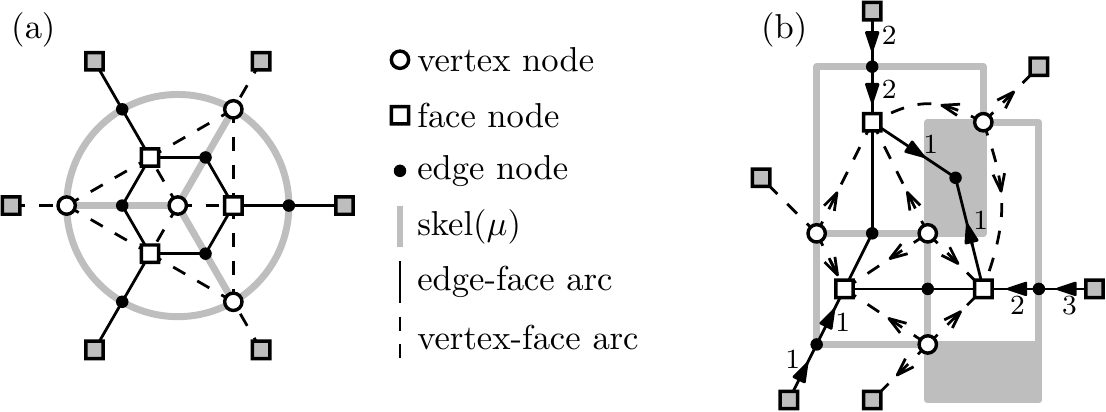}
  \caption{(a) The structure of the flow network $N^{\mathcal E}$ for
    the case that $\tau$ is an R-node with $\skel(\tau) = K_4$.  The
    outer face is split into several gray boxes to improve
    readability.  (b) A flow together with the corresponding
    orthogonal representation.  The numbers indicate the amount of
    flow on the arcs.  Undirected edges imply~0 flow, directed arcs
    without a number have flow~1.}
  \label{fig:flow-network}
\end{figure}

Now we are ready to define the flow network $N^{\mathcal E}$ on
$\skel(\tau)$ with respect to the fixed embedding~$\mathcal E$ of
$\skel(\tau)$; see Figure~\ref{fig:flow-network}(a) for an example.
For each vertex $v$, each virtual edge $\eps$ and each face~$f$ in
$\skel(\tau)$ the flow network $N^{\mathcal E}$ contains the nodes
$v$, $\eps$ and $f$, called \emph{vertex node}, \emph{edge node} and
\emph{face node}, respectively.  The network $N^\mathcal E$ contains
the arcs $(v, f)$ and $(f, v)$ with capacity~1, called
\emph{vertex-face arcs}, if the vertex $v$ and the face $f$ are
incident in $\skel(\tau)$.  For every virtual edge $\eps$ we add
\emph{edge-face arcs} $(\eps, f)$ and $(f, \eps)$, if $f$ is incident
to $\eps$.  We use $\cost_H(\cdot) - b_H$ as cost function of the
arc~$(f, \eps)$, where $H$ is the expansion graph of the virtual edge
$\eps$.  The edge-face arcs $(\eps, f)$ in the opposite direction have
infinite capacity with~0 cost.  It remains to define the demand of
every node in $N^{\mathcal E}$.  Every inner face has a demand of~4,
the outer face has a demand of~$-4$.  An edge node $\eps$ stemming
from the edge $\eps = \{s, t\}$ with expansion graph $H$ has a demand
of $\deg_H(s) + \deg_H(t) - 2$, where $\deg_H(v)$ denotes the degree
of $v$ in $H$.  The demand of a vertex node~$v$ is $4 - \deg_G(v) -
\deg_{\skel(\tau)}(v)$.

In the flow network $N^{\mathcal E}$ the flow entering a face node $f$
using a vertex-face arc or an edge-face arc is interpreted as the
rotation at the corresponding vertex or along the path between the
poles of the corresponding child, respectively; see
Figure~\ref{fig:flow-network}(b) for an example.  Incoming flow is
positive rotation and outgoing flow negative rotation.  Let $b_{H_1},
\dots, b_{H_k}$ be the base costs of the expansion graphs
corresponding to virtual edges in $\skel(\tau)$.  We define the
\emph{total base costs} of $\tau$ to be $b_\tau = \sum_i b_{H_i}$.
Note that the total base costs of $\tau$ are a lower bound for the
costs that have to be paid for every orthogonal representation of $G$.
We show that an optimal flow $\phi$ in $N^{\mathcal E}$ corresponds to
a $(\tau, \mathcal E)$-optimal orthogonal representation $\mathcal R$
of $G$.  Since the base costs do not appear in the flow network, the
costs of the flow and its corresponding orthogonal representation
differ by the total base costs $b_\tau$, that is $\cost(\mathcal R) =
\cost(\phi) + b_\tau$.  We obtain the following lemma.

\begin{lemma}
  \label{lem:compute-optimal-rep-root}
  Let $G$ be a biconnected convex instance of {\sc OptimalFlexDraw},
  let $\mathcal T$ be its SPQR-tree with root $\tau$ and let $\mathcal
  E$ be an embedding of $\skel(\tau)$.  If the cost function of every
  principal split component is known, a $(\tau, \mathcal E)$-optimal
  solution can be computed in $\mathcal O(T_{\flow}(|\skel(\tau)|))$
  time.
\end{lemma}
\begin{proof}
  As mentioned before, we want to use the flow network $N^{\mathcal
    E}$ to compute an optimal orthogonal representation.  To this end
  we show two directions.  First, given a $(\tau, \mathcal E)$-optimal
  orthogonal representation $\mathcal R$, we obtain a feasible flow
  $\phi$ in $N^{\mathcal E}$ such that $\cost(\phi) = \cost(\mathcal
  R) - b_\tau$, where $b_\tau$ are the total base costs.  Conversely,
  given an optimal flow $\phi$ in $N^{\mathcal E}$, we show how to
  construct an orthogonal representation $\mathcal R$ such that
  $\cost(\mathcal R) = \cost(\phi) + b_\tau$.  As the flow network
  $N^{\mathcal E}$ has size $\mathcal O(|\skel(\tau)|)$, the claimed
  running time follows immediately.

  Let $\mathcal R$ be a $(\tau, \mathcal E)$-optimal orthogonal
  representation of $G$.  As we only consider nice and thus only tight
  drawings we can assume the orthogonal representation $\mathcal R$ to
  be tight.  Recall that being tight implies that the poles of the
  expansion graph of every virtual edge have a rotation of~1 in the
  internal faces.  We first show how to assign flow to the arcs in
  $N^{\mathcal E}$.  It can then be shown that the resulting flow is
  feasible and causes $\cost(\mathcal R) - b_\tau$ cost.  For every
  pair of vertex-face arcs $(f, v)$ and~$(v, f)$ in $N^{\mathcal E}$
  there exists a corresponding face $f$ in the orthogonal
  representation $\mathcal R$ of $G$ and we set $\phi((v, f)) =
  \rot(v_f)$.  Let $\eps = \{s, t\}$ be a virtual edge in $\skel(\mu)$
  incident to the two faces $f_1$ and~$f_2$.  Without loss of
  generality let $\pi_{f_1}(s, t)$ be the path belonging to the
  expansion graph of $\eps$.  Then $\pi_{f_2}(t, s)$ also belongs to
  $H$.  We set $\phi((\eps, f_1)) = \rot_{\mathcal R}(\pi_{f_1}(s,
  t))$ and $\phi((\eps, f_2)) = \rot_{\mathcal R}(\pi_{f_2}(t, s))$.
  For the resulting flow $\phi$ we need to show that the capacity of
  every arc is respected, that the demand of every vertex is
  satisfied, and that $\cost(\phi) = \cost(\mathcal R) - b_\tau$
  holds.

  First note that the flow on the vertex-face arcs does not exceed the
  capacities of~1 since every vertex has degree at least~2.  Since no
  other arc has a capacity, it remains to deal with the demands and
  the costs.

  \begin{figure}
    \centering
    \includegraphics[page=2]{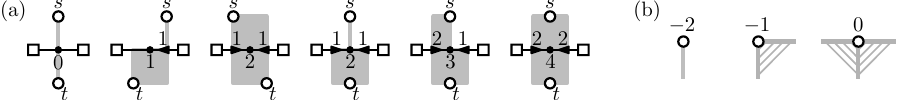}
    \caption{(a) Illustration of the demand of virtual edges.  (b)
      Rotation of poles in the outer face, depending on the degree.}
    \label{fig:flow-network-2}
  \end{figure}

  For the demands we consider each vertex type separately.  Let $f$ be
  a face node.  The total incoming flow entering $f$ is obviously
  equal to the rotation in $\mathcal R$ around the face $f$.  As
  $\mathcal R$ is an orthogonal representation this rotation equals
  to~4 ($-4$ for the outer face), which is exactly the demand of $f$.
  Let $\eps$ be an edge node corresponding to the expansion graph $H$
  with poles~$s$ and~$t$.  Recall that $\dem(\eps) = \deg_H(s) +
  \deg_H(t) - 2$ is the demand of $\eps$.
  Figure~\ref{fig:flow-network-2}(a) illustrates the demand of a
  virtual edge.  Let $\mathcal S$ be the orthogonal representation
  induced on $H$ by $\mathcal R$ and let $f$ be the outer face of
  $\mathcal S$.  Clearly, the flow leaving $\eps$ is equal to
  $\rot_{\mathcal R}(\pi_{f_1}(s, t)) + \rot_{\mathcal R}(\pi_{f_2}(t,
  s)) = \rot_{\mathcal S}(\pi_f(s, t)) + \rot_{\mathcal S}(\pi_f(t,
  s))$.  Since $f$ is the outer face of $H$, the total rotation around
  this faces sums up to~$-4$.  The rotation of the pole $s$ in the
  outer face $f$ is $\deg_H(s) - 3$, see
  Figures~\ref{fig:flow-network-2}(b), and the same holds for $t$.
  Thus we have $\rot_{\mathcal S}(\pi_f(s, t)) + \rot_{\mathcal
    S}(\pi_f(t, s)) + \deg_H(s) - 3 + \deg_H(t) - 3 = -4$.  This
  yields for the outgoing flow $\rot_{\mathcal S}(\pi_f(s, t)) +
  \rot_{\mathcal S}(\pi_f(t, s)) = 2 - \deg_H(s) - \deg_H(t)$, which
  is exactly the negative demand of $\eps$.  It remains to consider
  the vertex nodes.  Let $v$ be a vertex node, recall that $\dem(v) =
  4 - \deg_G(v) - \deg_{\skel(\tau)}(v)$ holds.  The outgoing flow
  leaving $v$ is equal to the summed rotation of $v$ in faces not
  belonging to expansion graphs of virtual edges in $\skel(\tau)$.
  As~$\mathcal R$ is an orthogonal representation, the total rotation
  around every vertex $v$ is $2 \cdot(\deg_G(v) - 2)$.  Moreover, $v$
  is incident to $\deg_{\skel(\tau)}(v)$ faces that are not contained
  in expansion graphs of virtual edges of $\skel(\tau)$.  Thus there
  are $\deg_G(v) - \deg_{\skel(\tau)}(v)$ faces incident to $v$
  belonging to expansion graphs.  As we assumed that the orthogonal
  representation of every expansion graph is tight, the rotation of
  $v$ in each of these faces is~1.  Thus the rotation of $v$ in the
  remaining faces not belonging to expansion graphs is $2
  \cdot(\deg_G(v) - 2) - (\deg_G(v) - \deg_{\skel(\tau)}(v))$.
  Rearrangement yields a rotation, and thus an outgoing flow, of
  $\deg_G(v) + \deg_{\skel(\tau)}(v) - 4$, which is the negative
  demand of $v$.

  To show that $\cost(\phi) = \cost(\mathcal R) - b_\tau$ holds it
  suffices to consider the flow on the edge-face arcs as no other arcs
  cause cost.  Let $\eps$ be a virtual edge and let $f_1$ and $f_2$
  the two incident faces.  The flow entering $f_1$ or $f_2$ does not
  cause any cost, as $(\eps, f_1)$ and $(\eps, f_2)$ have infinite
  capacity with~0 cost.  Thus only flow entering $\eps$ over the arcs
  $(f_1, \eps)$ and~$(f_2, \eps)$ may cause cost.  Assume without loss
  of generality that the number of bends $\rho$ the expansion graph
  $H$ of $\eps$ has is determined by the rotation along $\pi_{f_1}(s,
  t)$, that is $\rho = - \rot_{\mathcal R}(\pi_{f_1}(s, t))$.  Let
  $\rho' = - \rot_{\mathcal R}(\pi_{f_2}(t, s))$ be the negative
  rotation along the path $\pi_{f_2}(t, s)$ in the face $f_2$.  Note
  that $\phi((f_1, \eps)) = \rho$ and $\phi((f_2, \eps)) = \rho'$.
  Obviously, the flow on $(f_1, \eps)$ causes the cost $\cost_H(\rho)
  - b_H$.  We show that the cost caused by the flow on $(f_2, \eps)$
  is~0.  If $\rho' \le 0$ this is obviously true, as there is no flow
  on the edge $(f_2, \eps)$.  Otherwise, $0 < \rho' \le \rho$ holds.
  It follows that the smallest possible number of bends $\ell_H$ every
  orthogonal representation of $H$ has lies between $\rho'$ and
  $\rho$.  It follows from the definition of $\cost_H(\cdot)$ and from
  the fact that all cost functions are convex that $\cost_H(\rho') =
  b_H$.  To sum up, the total cost on edge-face arcs incident to the
  virtual edge $\eps$ is equal to the cost caused by its expansion
  graph $H$ with respect to the orthogonal representation $\mathcal R$
  minus the base cost $b_H$.  As neither $\phi$ nor $\mathcal R$ have
  additional cost we obtain $\cost(\phi) = \cost(\mathcal R) -
  b_\tau$.

  It remains to show the opposite direction, that is given an optimal
  flow $\phi$ in $N^{\mathcal E}$, we can construct an orthogonal
  representation $\mathcal R$ of $G$ such that $\cost(\mathcal R) =
  \cost(\phi) + b_\tau$.  This can be done by reversing the
  construction above.  The flow on edge-face arcs determines the
  number of bends for the expansion graphs of each virtual edge.  The
  cost functions of these expansion graphs guarantee the existence of
  orthogonal representations with the desired rotations along the
  paths between the poles, thus we can assume to have orthogonal
  representations for all children.  We combine these orthogonal
  representations by setting the rotations between them at common
  poles as specified by the flow on vertex-face arcs.  It can be
  easily verified that this yields an orthogonal representation of the
  whole graph $G$ by applying the above computation in the opposite
  direction.
\end{proof}

The above results rely on the fact that the cost functions of
principal split components are convex as stated in
Theorem~\ref{thm:convex-cost-function} and that they can be computed
efficiently.  In the following we show that
Theorem~\ref{thm:convex-cost-function} really holds with the help of a
structural induction over the SPQR-tree.  More precisely, the cost
functions of principal split components corresponding to the leaves of
$\mathcal T$ are the cost functions of the edges and thus they are
convex.  For an inner node $\mu$ we assume that the pertinent graphs
of the children of $\mu$ have convex cost functions and show that $H =
\pert(\mu)$ itself also has a convex cost function.  The proof is
constructive in the sense that it directly yields an algorithm to
compute these cost functions bottom up in the SPQR-tree.

Note that we can again apply Lemma~\ref{lem:contract-deg-33} in the
case that the cost function of the expansion graph of one of the
virtual edges in $\skel(\mu)$ is not convex due to the fact that both
of its poles have degree~3.  This means that we can simply contract
such a virtual edge (corresponding to a contraction of the expansion
graph in $H$), compute the cost function for the remaining graph
instead of $H$ and plug the contracted expansion graph into the
resulting orthogonal representations.  Thus we can assume that the
cost function of each of the expansion graphs is convex, without any
exceptions.

The flow network $N^{\mathcal E}$ that was introduced to compute an
optimal orthogonal representation in the root of the SPQR-tree can be
adapted to compute the cost function of the principal split component
$H$ corresponding to a non-root node $\mu$.  To this end we have to
deal with the parent edge, which does not occur in the root of
$\mathcal T$, and we consider a parameterization of $N^{\mathcal E}$
to compute several optimal orthogonal representations with a
prescribed number of bends, depending on the parameter in the flow
network.  Before we describe the changes in the flow network we need
to make some considerations about the cost function.  By the
definition of the cost function it explicitely optimizes over all
planar embeddings of $\skel(\mu)$.  Moreover, as the cost function
$\cost_H(\rho)$ depends on the number of bends $\rho$ a graph $H$ has,
it implicitly allows to flip the embedding of $H$ since the number of
bends is defined as $\max\{|\rot(\pi(s, t))|, |\rot(\pi(t, s))|\}$.
However, the flow network $N^{\mathcal E}$ can only be used to compute
the cost function for a fixed embedding.  Thus we define the
\emph{partial cost function} $\cost_H^{\mathcal E}(\rho)$ of $H$ with
respect to the planar embedding $\mathcal E$ of $\skel(\mu)$ to be the
smallest possible cost of an orthogonal representation inducing the
planar embedding $\mathcal E$ on $\skel(\mu)$ with $\rho$ bends such
that the number of bends is determined by $\pi_f(s, t)$, that is
$\rot(\pi_f(s, t)) = - \rho$, where~$f$ is the outer face.  Note that
the minimum over the partial cost functions $\cost_H^{\mathcal
  E}(\cdot)$ and $\cost_H^{\mathcal E'}(\cdot)$, where $\mathcal E'$
is obtained by flipping the embedding $\mathcal E$ of $\skel(\mu)$
yields a function describing the costs of $H$ with respect to the
embedding $\mathcal E$ of $\skel(\mu)$ depending on the number of
bends $H$ has (and not on the rotation along $\pi_f(s, t)$ as the
partial cost function does).  Obviously, minimizing over all partial
cost functions yields the cost function of $H$.

The flow network $N^{\mathcal E}$ is defined as before with the
following modifications.  The parent edge of $\skel(\mu)$ does not
have a corresponding edge node.  Let $f_1$ and $f_2$ be the faces in
$\skel(\mu)$ incident to the parent edge.  These two faces together
form the outer face $f$ of $H$, thus we could merge them into a single
face node.  However, not merging them has the advantage that the
incoming flow in $f_1$ and $f_2$ corresponds to the rotations along
$\pi_f(s, t)$ and $\pi_f(t, s)$, respectively (it might be the other
way round but we can assume this situation without loss of
generality).  Thus, we do not merge~$f_1$ and $f_2$, which enables us
to control the number of bends of $H$ by setting the demands of~$f_1$
and~$f_2$.  This is also the reason why we remove the vertex-face arcs
between the poles and the two faces~$f_1$ and~$f_2$.  Before we
describe how to set the demands of $f_1$ and $f_2$, we fit the demands
of the poles to the new situation.  As we only consider tight
orthogonal representations we know that the rotation at the poles $s$
and $t$ in all inner faces is~1.  Thus, we set $\dem(s) = 2 -
\deg_{\skel(\mu)}(s)$ and $\dem(t) = 2 - \deg_{\skel(\mu)}(t)$ as this
is the number of faces incident to $s$ and $t$, respectively, after
removing the vertex-face arcs to $f_1$ and $f_2$.  With these
modifications the only flow entering $f_1$ and $f_2$ comes from the
paths $\pi_f(s, t)$ and $\pi_f(t, s)$, respectively.  As the total
rotation around the outer face is~$-4$ and the rotation at the
vertices $s$ and $t$ is $\deg_H(s) - 3$ and $\deg_H(t) - 3$,
respectively, we have to ensure that $\dem(f_1) + \dem(f_2) = 2 -
\deg_H(s) - \deg_H(t)$.  As mentioned before, we assume without loss
of generality that $\pi_f(s, t)$ belongs to the face $f_1$ and
$\pi_f(t, s)$ belongs to $f_2$.  Then the incoming flow entering $f_1$
corresponds to $\rot(\pi_f(s, t))$ of an orthogonal representation.
We parameterize $N^{\mathcal E}$ with respect to the faces $f_1$ and
$f_2$ starting with $\dem(f_1) = 0$ and $\dem(f_2) = 2 - \deg_H(s) -
\deg_H(t)$.  It obviously follows that an optimal flow in $N^{\mathcal
  E}$ with respect to the parameter $\rho$ corresponds to an optimal
orthogonal representation of $H$ that induces $\mathcal E$ on
$\skel(\mu)$ and has a rotation of $-\rho$ along~$\pi_f(s, t)$.  Thus,
up to the total base costs $b_\mu$, the cost function of the flow
network equals to the partial cost function of $H$ on the interval
$[\ell_H, 3]$, that is $\cost_{N^{\mathcal E}}(\rho) + b_\mu =
\cost_H^{\mathcal E}(\rho)$ for $\ell_H \le \rho \le 3$.  To obtain
the following lemma it remains to show two things for the case that
$\deg(s) + \deg(t) < 6$.  First, $\cost_{N^{\mathcal E}}(\rho)$ and
thus each partial cost function is convex for $\ell_H \le \rho \le 3$.
Second, the minimum over these partial cost functions is convex.

\begin{lemma}
  \label{lem:convex-cost-induction}
  If Theorem~\ref{thm:convex-cost-function} holds for each principal
  split component corresponding to a child of the node $\mu$ in the
  SPQR-tree, then it also holds for $\pert(\mu)$.
\end{lemma}
\begin{proof}
  As mentioned before, we can use the flow network $N^{\mathcal E}$ to
  compute the partial cost function $\cost_H^{\mathcal E}(\rho)$ for
  $\ell_H \le \rho \le 3$ since $\cost_H^{\mathcal E}(\rho) =
  \cost_{N^{\mathcal E}}(\rho) + b_\mu$ holds on this interval.  In
  the following we only consider the case where $\deg_H(s) + \deg_H(t)
  < 6$ holds for the poles $s$ and $t$.  For the case $\deg_H(s) =
  \deg_H(t) = 3$ we do not need to show anything.  To show that the
  partial cost function is convex we do the following.  First, we show
  that $\cost_H^{\mathcal E}(\rho)$ is minimal for $\rho = \ell_H$.
  This implies that the cost function $\cost_{N^{\mathcal E}}(\rho)$
  of the flow network is minimal for $\rho = \rho_0 \le \ell_H$.  Then
  Theorem~\ref{thm:flow-has-convex-cost} can be applied showing that
  $\cost_{N^{\mathcal E}}(\rho)$ is convex for $\rho \in [\rho_0,
  \infty]$ yielding that the partial cost function $\cost_H^{\mathcal
    E}(\rho)$ is convex for $\rho \in [\ell_H, 3]$.  Thus, it remains
  to show that $\cost_H^{\mathcal E}(\rho)$ is minimal for $\rho =
  \ell_H$ to obtain convexity for the partial cost functions.

  Let $\mathcal S$ be an orthogonal representation of $H$ with $\rho
  \in [\ell_H, 3]$ bends such that $\pi_f(s, t)$ determines the number
  of bends, that is $\rot_{\mathcal S}(\pi_f(s, t)) = - \rho$, where
  $f$ is the outer face of $H$.  We show the existence of an
  orthogonal representation $S'$ with $\rot_{\mathcal S'}(\pi_f(s, t))
  = -\ell_H$ and $\cost(S') \le \cost(S)$.  Since we assume $\mathcal
  S$ to be tight, the rotations at the poles $\rot_{\mathcal S}(s_f)$
  and $\rot_{\mathcal S}(t_f)$ only depend on the degree of $s$ and
  $t$.  More precisely, we have $\rot_{\mathcal S}(s_f) = \deg_H(s) -
  3$ and the same holds for $t$.  Since the total rotation around the
  outer face $f$ is~$-4$ the following equation holds.
  \begin{equation}
    \rot_{\mathcal S}(\pi_f(t, s)) = \rho +  2 - \deg_H(s) -
    \deg_H(t)\label{eq:1}
  \end{equation}
  In the following we show that $\rot_{\mathcal S}(\pi_f(t, s)) \ge 0$
  holds if the number of bends $\rho$ exceeds $\ell_H$.  Then
  Corollary~\ref{cor:rotations-form-interval} in combination with
  Theorem~\ref{thm:valid-and-optimal} can be used to reduce the
  rotation along $\pi_f(t, s)$ and thus reduce the number of bends
  by~1, yielding finally an orthogonal representation with $\ell_H$
  bends determined by $\pi_f(s, t)$.  Recall that the lower bound for
  the number of bends was defined as $\ell_H = \lceil(\deg(s) +
  \deg(t) - 2)/2\rceil$.  First consider the case that $\deg_H(s) +
  \deg_H(t)$ is even (and of course less than~6).  Then
  Equation~\eqref{eq:1} yields $\rot_{\mathcal S}(\pi_f(t, s)) = \rho
  -2\ell_H$.  If $\rho$ is greater than $\ell_H$ this yields
  $\rot_{\mathcal S}(\pi_f(t, s)) > -\ell_H$.  Since $\ell_H$ is at
  most~$1$ in the case that $\deg(s) + \deg(t)$ is even and less
  than~6, this yields $\rot_{\mathcal S}(\pi_f(t, s)) > -1$.  The case
  that $\deg_H(s) + \deg_H(t)$ is odd works similarly.  Then
  Equation~\eqref{eq:1} yields $\rot_{\mathcal S}(\pi_f(t, s)) = \rho
  - 2\ell_H + 1$.  As before $\rho$ is assumed to be greater
  than~$\ell_H$ yielding $\rot_{\mathcal S}(\pi_f(t, s)) > -\ell_H +
  1$.  As $\ell_H$ is at most~2 we again obtain $\rot_{\mathcal
    S}(\pi_f(t, s)) > -1$, which concludes the proof that the partial
  cost functions are convex.

  It remains to show that the minimum over the partial cost functions
  is convex.  First assume that $\mu$ is an R-node.  Then its skeleton
  has only two embeddings $\mathcal E$ and $\mathcal E'$ where
  $\mathcal E'$ is obtained by flipping $\mathcal E$.  We have to show
  that the minimum over the two partial cost functions
  $\cost_H^{\mathcal E}(\cdot)$ and $\cost_H^{\mathcal E'}(\cdot)$
  remains convex.  For the case that $\deg(s) + \deg(t) = 5$ the
  equation $\ell_H = 2$ holds and thus we only have to show convexity
  on the interval $[2, 3]$.  Obviously, $\cost_H(\cdot)$ is convex on
  this interval if and only if $\cost_H(2) \le \cost_H(3)$.  As this
  is the case for both partial cost functions, it is also true for the
  minimum.  For $\deg(s) + \deg(t) < 5$ we first show that
  $\cost_H^{\mathcal E}(\ell_H) = \cost_H^{\mathcal E'}(\ell_H)$
  holds.  For the case that $\deg(s) + \deg(t)$ is even this is clear
  since mirroring an orthogonal representation~$\mathcal S$ with
  $\rot_{\mathcal S}(\pi_f(s, t)) = -\ell_H$ inducing $\mathcal E$ on
  $\skel(\mu)$ yields an orthogonal representation $\mathcal S'$ with
  $\rot_{\mathcal S'}(\pi_f(s, t)) = -\ell_H$ inducing $\mathcal E'$
  on $\skel(\mu)$.  For the case that $\deg(s) + \deg(t) = 3$, the
  orthogonal representation $\mathcal S$ with rotation~$-1$ along
  $\pi_f(s, t)$ can also be mirrored yielding $\mathcal S'$ with
  rotation~0 along $\pi_f(s, t)$.  By
  Corollary~\ref{cor:rotations-form-interval} this rotation can be
  reduced to~$-1$ without causing any additional cost.  As this
  construction also works in the opposite direction we have
  $\cost_H^{\mathcal E}(\ell_H) = \cost_H^{\mathcal E'}(\ell_H)$ for
  all cases.  Moreover, $\cost_H^{\mathcal E}(0) = \cost_H^{\mathcal
    E}(1)$ holds by definition, if $\deg(s) + \deg(t) > 2$.  If
  $\deg(s) = \deg(t) = 1$ this equation is also true as the rotation
  along $\pi_f(s, t)$ of an orthogonal representation can be reduced
  by~1 if it is~0, again due to
  Corollary~\ref{cor:rotations-form-interval}.  Thus it remains to
  show that the cost function $\cost_H(\cdot)$ defined as the minimum
  of $\cost_H^{\mathcal E}(\cdot)$ and $\cost_H^{\mathcal E'}(\cdot)$
  is convex on the interval $[1, 3]$.

  Assume for a contradiction that $\cost_H(\rho)$ is not convex for
  $\rho \in [1, 3]$, that is $\Delta\cost_H(1) > \Delta\cost_H(2)$.
  Assume without loss of generality that $\cost_H(3) =
  \cost_H^{\mathcal E}(3)$ holds.  As we showed before $\cost_H(1) =
  \cost_H^{\mathcal E}(1)$ also holds.  Since $\cost_H(2)$ is the
  minimum over $\cost_H^{\mathcal E}(2)$ and $\cost_H^{\mathcal E}(2)$
  we additionally have $\cost_H(2) \le \cost_H^{\mathcal E}(2)$.  This
  implies that the inequalities $\Delta\cost_H^{\mathcal E}(1) \ge
  \Delta\cost_H(1)$ and $\Delta\cost_H^{\mathcal E}(2) \le
  \Delta\cost_H(2)$ hold, yielding that the partial cost function
  $\cost_H^{\mathcal E}(\rho)$ is not convex for $\rho \in [1, 3]$,
  which is a contradiction.  Thus $\cost_H(\cdot)$ is convex.

  The case that $\mu$ is a P-node works similar to the case that $\mu$
  is an R-node.  If $\mu$ has only two children, its skeleton has only
  two embeddings $\mathcal E$ and $\mathcal E'$ obtained from one
  another by flipping.  Thus the same argument as for R-nodes applies.
  If $\mu$ has three children, then $\deg(s) = \deg(t) = 3$ holds and
  thus we do not have to show convexity.  Note that in the case
  $\deg(s) = \deg(t) = 3$ the resulting cost function can be computed
  by taking the minimum over the partial cost functions with respect
  to all embeddings of $\skel(\mu)$, although it may by non-convex.
  If $\mu$ is an S-node, we have a unique embedding and thus the
  partial cost function with respect to this embedding is already the
  cost function of $H$.  Note that considering only the rotation along
  $\pi_f(s, t)$ for the partial cost function is not a restriction, as
  S-nodes are completely symmetric.
\end{proof}

Lemma~\ref{lem:convex-cost-induction} together with the fact that the
cost function of every edge is convex shows that
Theorem~\ref{thm:convex-cost-function} holds, that is the cost
functions of all principal split components are convex on the
interesting interval $[0, 3]$ except for the special case where both
poles have degree~3.  However, this special case is easy to handle as
principal split components of this type with non-convex cost functions
can be simply contracted to a single vertex by
Lemma~\ref{lem:contract-deg-33}.  Moreover, the proof is constructive
in the sense that it shows how the cost functions can be computed
efficiently bottom up in the SPQR-tree.  For each node $\mu$ we have
to solve a constant number of minimum-cost flow problems in a flow
network of size $\mathcal O(|\skel(\mu)|)$.  As the total size of all
skeletons in $\mathcal T$ is linear in the number $n$ of vertices
in~$G$, we obtain an overall $\mathcal O(T_{\flow}(n))$ running time
to compute the cost functions with respect to the root $\tau$.
Finally, Lemma~\ref{lem:compute-optimal-rep-root} can be applied to
compute an optimal orthogonal representation with respect to a fixed
root and a fixed embedding of the root's skeleton in $\mathcal
O(T_{\flow}(|\skel(\tau)|))$ time.  To compute an overall optimal
solution, we have to compute a $(\tau, \mathcal E)$-optimal solution
for every root $\tau$ and every embedding $\mathcal E$ of
$\skel(\tau)$.  The number of embeddings of $\skel(\tau)$ is linear in
the size of $\skel(\tau)$ (since P-nodes have at most degree~4) and
the total size of all skeletons is linear in $n$.  We obtain the
following theorem.

\begin{theorem}
  \label{thm:solve-opti-flex-biconn}
  {\sc OptimalFlexDraw} can be solved in $\mathcal O(n \cdot
  T_{\flow}(n))$ time for convex biconnected instances.
\end{theorem}

\subsection{Connected Graphs}
\label{sec:connected-graphs}

In this section we extend the result obtained in
Section~\ref{sec:biconnected-graphs} to the case that the input graph
$G$ contains cutvertices.  Let $\mathcal B$ be the BC-tree of $G$
rooted at some B-node $\beta$.  Then every Block except for $\beta$
has a unique cutvertex as parent and we need to find optimal
orthogonal representations with the restriction that this cutvertex
lies on the outer face.  We claim that we can then combine these
orthogonal representations of the blocks without additional cost.

Unfortunately, with the so far presented results we cannot compute the
optimal orthogonal representation of a biconnected graph considering
only embeddings where a specific vertex $v$ lies on the outer face.
We may restrict the embeddings of the skeletons we consider when
traversing the SPQR-tree bottom up to those who have $v$ on the outer
face.  However, we can then no longer assume that the cost functions
we obtain are symmetric.  To deal with this problem, we present a
modification of the SPQR-tree, that can be used to represent exactly
the planar embeddings that have $v$ on the outer face and are
represented by the SPQR-tree rooted at a node $\tau$.

Let $\tau$ be the root of the SPQR-tree $\mathcal T$.  If $v$ is a
vertex of $\skel(\tau)$, then restricting the embeddings of
$\skel(\tau)$ to those who have $v$ on the outer face of $\skel(\tau)$
forces $v$ to be on the outer face of the resulting embedding of $G$.
Otherwise, $v$ is contained in the expansion graph of a unique virtual
edge $\eps$ in $\skel(\tau)$, we say that $v$ is \emph{contained} in
$\eps$.  Obviously, $\eps$ has to be on the outer face of the
embedding of $\skel(\tau)$.  However, this is not sufficient and it
depends on the child $\mu$ of $\tau$ corresponding to $\eps$ whether
$v$ lies on the outer face of the resulting embedding of $G$.  Let
$\mathcal E_\tau$ be an embedding of $\skel(\tau)$ having $\eps$ on
the outer face and let $s$ and $t$ be the endpoints of $\eps$.  Then
there are two possibilities, either $\eps = \{s, t\}$ has the outer
face to the left or to the right, where the terms ``left'' and
``right'' are with respect to an orientation from $t$ to $s$.  Assume
without loss of generality that the outer face lies to the right of
$\eps$ and consider the child $\mu$ of $\tau$ corresponding to $\eps$.
As $\mathcal T$ is rooted, we consider only embeddings of $\skel(\mu)$
that have the parent edge $\{s, t\}$ on the outer face.  As the choice
of the outer face of $\skel(\mu)$ does not have any effect on the
resulting embedding, we can assume that $\{s, t\}$ lies to the left of
$\skel(\mu)$, that is the inner face incident to $\{s, t\}$ lies to
the right of $\{s, t\}$ with respect to an orientation from $t$ to
$s$.  A vertex contained in $\skel(\mu)$ then lies obviously on the
outer face of the resulting embedding of $G$ if and only if it lies on
the outer face of the embedding of $\skel(\mu)$.  Thus, if $v$ is
contained in $\skel(\mu)$, restricting the embedding choices such that
$v$ lies on the outer face of $\skel(\mu)$ forces $v$ to be on the
outer face of $G$.  Note that in this case $\mu$ is either an R- or an
S-node.  For S-nodes there is no embedding choice and every vertex in
$\skel(\mu)$ lies on the outer face in this embedding.  If $\mu$ is an
R-node, there are only two embeddings and either $v$ lies on the outer
face of exactly one of them or in none of them.  In the latter case
the SPQR-tree with respect to the root $\tau$ does not represent an
embedding of $G$ with $v$ on the outer face at all.

Assume that $v$ is not contained in $\skel(\mu)$.  Then it is again
contained in a single virtual edge $\eps'$ and it is necessary that
$\eps'$ lies on the outer face of the embedding of $\skel(\mu)$.
Moreover, it depends on the child of $\mu$ corresponding to $\eps'$
whether $v$ really lies on the outer face.  Note that fixing $\eps'$
on the outer face completely determines the embedding of $\skel(\mu)$
if it is not a P-node.  If $\mu$ is a P-node, the virtual edge $\eps'$
has to be the rightmost, whereas the order of all other virtual edges
can be chosen arbitrarily.  If this is the case we split the P-node
into two parts, one representing the fixed embedding of $\eps'$, the
other representing the choices for the remaining edges; see
Figure~\ref{fig:spqr-tree-vertex-outer-face}(a).  More precisely, we
split $\mu$ into two P-nodes, the first one containing the parent edge
$\{s, t\}$, the edge $\eps'$ and a new virtual edge corresponding to
the second P-node, which is inserted as child.  The skeleton of the
second P-node contains a parent edge corresponding to the first P-node
and the remaining virtual edges that were contained in $\skel(\mu)$
but are not contained in the first P-node.  The children of $\mu$ are
attached to the two P-nodes depending on where the corresponding
virtual edges are.  Note that by splitting the P-node $\mu$, the
virtual edge $\eps'$ can no longer be in between two other virtual
edges in $\mu$.  However, this is a required restriction, thus we do
not loose embeddings that we want to represent.  Moreover, the new
P-node containing the virtual edge $\eps'$ that need to be fixed to
the outer face contains only two virtual edges (plus the parent edge)
and thus the embedding of its skeleton is completely fixed by
requiring $\eps'$ to be on the outer face.

\begin{figure}
  \centering
  \includegraphics{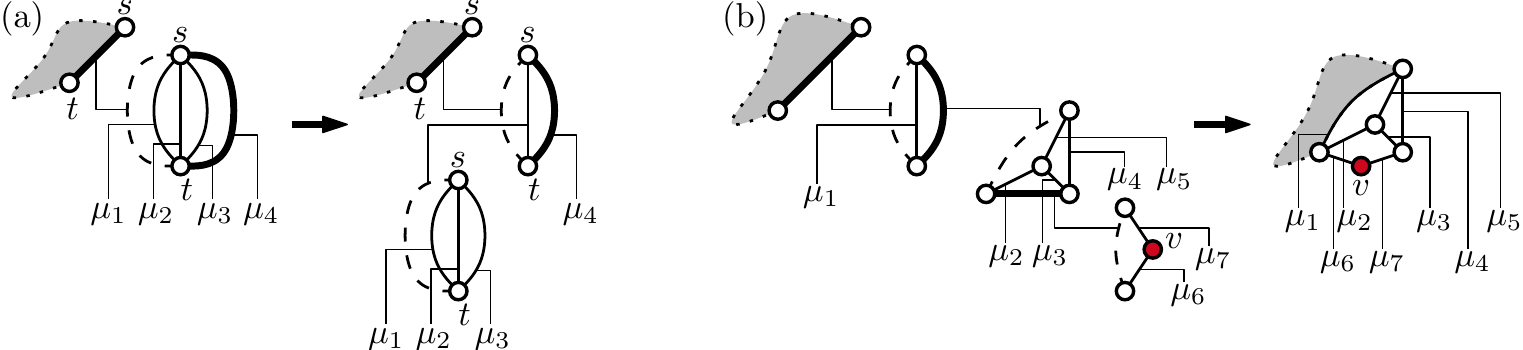}
  \caption{(a) Splitting a P-node into two P-nodes, the vertex $v$
    fixed to the outer face is contained in the thick edges.  (b)
    Contracting the path from the root to the node containing $v$ in
    its skeleton.}
  \label{fig:spqr-tree-vertex-outer-face}
\end{figure}

To sum up, if $\skel(\tau)$ contains $v$, then we simply have to
choose an embedding of $\skel(\tau)$ with~$v$ on the outer face.
Otherwise, we have to fix the virtual edge containing $v$ to the outer
face and additionally have to consider the child of $\tau$
corresponding to this virtual edge.  For the child we then have
essentially the same situation.  Either $v$ is contained in its
skeleton, then the embedding is fixed to the unique embedding having
$v$ on the outer face or $v$ is contained in some virtual edge.
However, then the embedding of the skeleton is again completely fixed
(P-nodes have to be split up first) and we can continue with the child
corresponding to the virtual edge containing $v$.  This yields a path
of nodes starting with the root $\tau$ having a completely fixed
embedding only depending on the embedding $\mathcal E_\tau$ chosen for
$\skel(\tau)$.  As the nodes on the path do not represent any
embedding choices, we can simply contract the whole path into a single
new root node, merging the skeletons on the path, such that the
embedding of the new skeleton of the root is still fixed.  This
contraction is illustrated in
Figure~\ref{fig:spqr-tree-vertex-outer-face}(b).  More precisely, let
$\tau$ be the root and let $\eps$ be the edge containing $v$,
corresponding to the child $\mu$.  Then we merge $\tau$ and $\mu$ by
replacing $\eps$ in $\tau$ by the skeleton of $\mu$ without the parent
edge.  The children of~$\mu$ are of course attached to the new root
$\tau'$ since $\skel(\tau')$ contains the corresponding virtual edges.
As mentioned before, the embedding of $\skel(\mu)$ was fixed by the
requirement that $v$ is on the outer face, thus the new skeleton
$\skel(\tau')$ has a unique embedding $\mathcal E_{\tau'}$ inducing
$\mathcal E_\tau$ on $\skel(\tau)$ and having $v$ or the new virtual
edge containing $v$ on the outer face.  The procedure of merging the
root with the child corresponding to the virtual edge containing $v$
is repeated until $v$ is contained in the skeleton of the root.  We
call the resulting tree the \emph{restricted SPQR-tree} with respect
to the vertex $v$ and to the embedding $\mathcal E_\tau$ of the root.

To come back to the problem {\sc OptimalFlexDraw}, we can easily apply
the algorithm presented in Section~\ref{sec:biconnected-graphs} to the
restricted SPQR-tree.  All nodes apart from the root are still S-, P-,
Q- or R-nodes and thus the cost functions with respect to the
corresponding pertinent graphs can be computed bottom up.  The root
$\tau$ may have a more complicated skeleton, however, its embedding is
fixed, thus we can apply the flow algorithm as before, yielding an
optimal drawing with respect to the chosen root $\tau$ and to the
embedding $\mathcal E_\tau$ of $\skel(\tau)$ with the additional
requirement that $v$ lies on the outer face.  Since the
\emph{restricted SPQR-tree} can be easily computed in linear time for
a chosen root $\tau$ and a fixed embedding $\mathcal E$ of
$\skel(\tau)$, we can compute a $(\tau, \mathcal E)$-optimal orthogonal
representation with the additional requirement that $v$ lies on the
outer face in $T_{\flow}(n)$ time, yielding the following theorem.

\begin{theorem}
  \label{thm:solve-opti-flex-biconn-v-outer}
  {\sc OptimalFlexDraw} with the additional requirement that a
  specific vertex lies on the outer face can be solved in $\mathcal
  O(n\cdot T_{\flow}(n))$ time for convex biconnected instances.
\end{theorem}

As motivated before, we can use the BC-tree to solve {\sc
  OptimalFlexDraw} for instances that are not necessarily biconnected.
We obtain the following theorem.

\begin{theorem}
  \label{thm:solve-opti-flex}
  {\sc OptimalFlexDraw} can be solved in $\mathcal O(n^2 \cdot
  T_{\flow}(n))$ time for convex instances.
\end{theorem}
\begin{proof}
  Let $G$ be a convex instance with positive flexibility of {\sc
    OptimalFlexDraw} and let $\mathcal B$ be its BC-tree rooted at
  some B-node $\beta$.  We show how to find an optimal drawing of $G$,
  optimizing over all embeddings represented by $\mathcal B$ with
  respect to the root $\beta$.  Then we can simply choose every B-node
  in~$\mathcal B$ to be the root once, solving {\sc OptimalFlexDraw}.
  The algorithm consumes $\mathcal O(n \cdot T_{\flow}(n))$ time for
  each root $\beta$ and thus the overall running time is $\mathcal
  O(n^2 \cdot T_{\flow}(n))$.  For the block corresponding to the root
  $\beta$ we use Theorem~\ref{thm:solve-opti-flex-biconn} to find the
  optimal orthogonal representation.  For all other blocks we use
  Theorem~\ref{thm:solve-opti-flex-biconn-v-outer} to find the optimal
  orthogonal representation with the cutvertex corresponding to the
  parent in $\mathcal B$ on the outer face.  It remains to stack these
  orthogonal representations together without causing additional cost.
  This can be easily done, if a cutvertex that is forced to lie on the
  outer face has all free incidences in the outer face and every other
  cutvertex has all free incidences in a single face.  The former can
  be achieved as we can assume orthogonal representations to be tight.
  If the latter condition is violated by a cutvertex $v$, then~$v$ has
  two incident edges $e_1$ and $e_2$ and the rotation of $v$ is~0 in
  both incident faces.  If both edges $e_1$ and $e_2$ have zero bends,
  we bend along a cycle around $v$ in the flex graph and thus we can
  assume without loss of generality that $e_1$ has a bend.  Moving $v$
  along $e_1$ to this bend yields an orthogonal representation where
  $v$ has both free incidences in the same face.  Thus given the
  orthogonal representations for the blocks, we can simply stack them
  together without causing additional cost.
\end{proof}

\subsection{Computing the Flow}
\label{sec:computing-flow}

In the previous sections we used $T_{\flow}(n)$ as placeholder for the
time necessary to compute a minimum-cost flow in a flow network of
size $n$.  Most minimum-cost flow algorithms do not consider the case
of multiple sinks and sources.  However, this is not a real problem as
we can simply add a \emph{supersink} connected to all sinks and a
\emph{supersource} connected to all sources.  Unfortunately, the
resulting flow network is no longer planar.  Orlin gives a strongly
polynomial time minimum-cost flow algorithm with running time
$\mathcal O(m \log n (m + n\log n))$, where $n$ is the number of
vertices and~$m$ the number of arcs~\cite{o-fspmcfa-93}.  Since our
flow network is planar (plus supersink and supersource) the number of
arcs is linear in the number of nodes.  Thus with this flow algorithm
we have $T_{\flow}(n) \in \mathcal O(n^2 \log^2 n)$.

Cornelsen and Karrenbauer give a minimum-cost flow algorithm for
planar flow networks with multiple sources and sinks consuming
$\mathcal O(\sqrt{\chi} \, n \log^3 n)$ time~\cite{ck-abm-12}, where
$\chi$ is the cost of the resulting flow.  Since the cost functions in
an instance of {\sc OptimalFlexDraw} may define exponentially large
costs in the size of the input, we cannot use this flow algorithm in
general to obtain a polynomial time algorithm.  However, in practice
it does not really make sense to have exponentially large costs.
Moreover, in several interesting special cases an optimal solution has
cost linear in the number of vertices.  We obtain the following
results.

\begin{corollary}
  A convex instance $G$ of {\sc OptimalFlexDraw} can be solved in
  $\mathcal O(n^4 \log^2 n)$ and $\mathcal O(\sqrt{\chi} \, n^3 \log^3
  n)$ time, where $\chi$ is the cost of an optimal solution.  The
  running time can be improved by a factor of $\mathcal O(n)$ for
  biconnected graphs.
\end{corollary}


\section{Conclusion}
\label{sec:conclusion}

We presented an efficient algorithm for the problem {\sc
  OptimalFlexDraw} that can be seen as the optimization problem
corresponding to {\sc FlexDraw}.  As a first step, we considered
biconnected 4-planar graphs with a fixed embedding and showed that
they always admit a nice drawing, which implies at most three bends
per edge except for a single edge on the outer face with up to four
bends.

Our algorithm for optimizing over all planar embeddings requires that
the first bend on every edge does not cause any cost as the problem
becomes $\mathcal{NP}$-hard otherwise.  Apart from that restriction we
allow the user to specify an arbitrary convex cost function
independently for each edge.  This enables the user to control the
resulting drawing.  For example, our algorithm can be used to minimize
the total number of bends, neglecting the first bend of each edge.
This special case is the natural optimization problem arising from the
decision problem {\sc FlexDraw}.  As another interesting special case,
one can require every edge to have at most two bends and minimize the
number of edges having more than one bend.  This enhances the
algorithm by Biedl and Kant~\cite{bk-bhogd-98} generating drawings
with at most two bends per edge with the possibility of optimization.
Note that in both special cases the cost of an optimal solution is
linear in the size of the graph, yielding a running time in $\mathcal
O(n^{\frac 7 2} \log^3 n)$ ($\mathcal O(n^{\frac 5 2} \log^3 n)$ if
the graph is biconnected).

\bibliographystyle{abbrv}
\bibliography{OptiFlex}

\end{document}